\documentclass[10pt, journal]{IEEEtran}

\usepackage{mathrsfs}
\usepackage{graphics} 

\usepackage{amsmath} 
\usepackage{amssymb}  

\usepackage{amsthm}

\usepackage{cite}
\usepackage{bm}
\usepackage{acronym}
\usepackage{paralist}
\usepackage{float}
\usepackage{color}
\usepackage{epstopdf}
\usepackage{multicol}
\usepackage{tikz}
\usepackage{hyperref}
\usepackage{graphicx}
\usepackage{mathtools}
\usepackage{soul}
\usepackage[bottom]{footmisc}

\usepackage{amsmath} 
\usepackage{amssymb}  
\usepackage{graphicx}
\usepackage{epstopdf}
\usepackage{amsmath}
\usepackage{mathtools}
\usepackage{dsfont}
\usepackage{tikz}
\usepackage{siunitx}
\usepackage{xcolor}
\usepackage[ruled,vlined,linesnumbered,noend]{algorithm2e}

\makeatletter
\hypersetup{colorlinks=true}
\AtBeginDocument{\@ifpackageloaded{hyperref}
  {\def\@linkcolor{blue}
  \def\@anchorcolor{red}
  \def\@citecolor{blue}
  \def\@filecolor{red}
  \def\@urlcolor{black}
  \def\@menucolor{red}
  \def\@pagecolor{red}
\begingroup
  \@makeother\`%
  \@makeother\=%
  \edef\x{%
    \edef\noexpand\x{%
      \endgroup
      \noexpand\toks@{%
        \catcode 96=\noexpand\the\catcode`\noexpand\`\relax
        \catcode 61=\noexpand\the\catcode`\noexpand\=\relax
      }%
    }%
    \noexpand\x
  }%
\x
\@makeother\`
\@makeother\=
}{}}
\makeatother

\newtheorem{Theorem}{Theorem}

\newtheorem{Lemma}{Lemma}

\newtheorem{Problem}{Problem}

\newtheorem{Remark}{Remark}

\newtheorem{Corollary}{Corollary}

\newtheorem{Assumption}{Assumption}

\newtheorem{Definition}{Definition}

\DeclareMathOperator*{\arginf}{arg\,inf}
\DeclareMathOperator*{\argsup}{arg\,su
p}

\IEEEoverridecommandlockouts

\def\BibTeX{{\rm B\kern-.05em{\sc i\kern-.025em b}\kern-.08em
    T\kern-.1667em\lower.7ex\hbox{E}\kern-.125emX}}

\begin{document}

\title{\LARGE{\bf Control Barrier Function based Attack-Recovery with Provable Guarantees}}

\author{Kunal~Garg, \IEEEmembership{Member}, \and Ricardo~G.~Sanfelice, \IEEEmembership{Fellow} \and
Alvaro~A.~Cardenas, \IEEEmembership{Senior Member}
\thanks{K.~Garg is with the School for Engineering of Matter, Transport and Energy at Arizona State University, R.G.~Sanfelice is with the Department of Electrical and Computer Engineering, and A.A.~Cardenas is with the Department of Computer Science and Engineering, University of California, Santa Cruz, CA, 95064, USA. E-mail(s): \texttt{kgarg24@asu.edu}, \texttt{\{ricardo, alacarde\}@ucsc.edu}.
Research partially supported by NSF Grants no. CNS-2039054 and CNS-2111688, by AFOSR Grants nos. FA9550-23-1-0145, FA9550-23-1-0313, and FA9550-23-1-0678, by AFRL Grant nos. FA8651-22-1-0017 and FA8651-23-1-0004, by ARO Grant no. W911NF-20-1-0253, DoD Grant no. W911NF-23-1-0158, AFOSR FA9550-24-1-0015, and the National Center for Transportation Cybersecurity and Resiliency (TraCR) (a U.S. Department of Transportation National
University Transportation Center) USDOT Grant \#69A3552344812. 
The views and conclusions contained in this document are those of the authors and should not be interpreted as representing the official policies, either expressed or implied, of the Army Research Office or the U.S. Government. The U.S. Government is authorized to reproduce and distribute reprints for Government purposes notwithstanding any copyright notation herein.
}

}
\maketitle
\thispagestyle{empty}

\begin{abstract}
This paper studies provable security guarantees for cyber-physical systems (CPS) under actuator attacks. In particular, we consider CPS safety and propose a new attack detection mechanism based on zeroing control barrier function (ZCBF) conditions. In addition, we design an adaptive recovery mechanism based on how close the system is to violating safety. We show that under certain conditions, the attack-detection mechanism is sound, i.e., there are no false negatives for adversarial attacks. We propose sufficient conditions for the initial conditions and input constraints so that the resulting CPS is secure by design. We also propose a novel hybrid control to account for attack detection delays and avoid Zeno behavior. Next, to efficiently compute the set of initial conditions, we propose a sampling-based method to verify whether a set is a viability domain. Specifically, we devise a method for checking a modified barrier function condition on a finite set of points to assess whether a set can be rendered forward invariant. Then, we propose an iterative algorithm to compute the set of initial conditions and input constraints set to limit the effect of an adversary if it compromises vulnerable inputs. Finally, we use a Quadratic Programming (QP) approach for online recovery (as well as nominal) control synthesis. We demonstrate the effectiveness of the proposed method in a simulation case study involving a quadrotor with an attack on its motors. 
\end{abstract}

\section{Introduction}

\subsection{Motivation}
Cyber-physical systems (CPS) such as autonomous and semi-autonomous air, ground, and space vehicles must maintain their safe operation and achieve mission objectives under various adversarial environments, including cyber-attacks. 
Security measures can be classified into two types of mechanisms~\cite{cardenascyber}: i) proactive, which considers design choices implemented in CPS \emph{before} attacks, and ii) reactive, which takes effect after an attack is detected. A proactive method, which considers design choices deployed in the CPS \emph{before} attacks, can result in a conservative design. However, reactive methods, which take effect after an attack is detected, heavily rely on fast and accurate attack-detection mechanisms. There is a plethora of work on attack detection for CPS, see, e.g.,~\cite{chen2018learning,choi2018,fengsystematic,renganathan2020distributionally}. However, as discussed in~\cite{urbina2016limiting}, a knowledgeable attacker can design stealthy attacks that can disrupt the nominal system behavior slowly to avoid these detection mechanisms. Such methods can cause system failure by pushing the system beyond its safe operating limits. An optimal approach to achieving resilience against cyber attacks must utilize the benefits of the two approaches while minimizing their limitations. 


\textit{Safety}, i.e., the system does not go out of a safe zone, is an essential requirement, violation of which can result in loss of money or human life, particularly when a system is under attack \cite{al2018cyber}. In most practical problems involving CPS, safety can be realized as guaranteeing the forward invariance of a safe set. 
One of the most common approaches to ensure that system trajectories remain in a safe set or that the safe set is forward invariant is based on a control barrier function (CBF), as it allows for a real-time implementable quadratic programming (QP)-based control synthesis framework \cite{ames2017control,garg2021sampling}. 

Most of the previous work on safety using CBFs, e.g., \cite{ames2017control}, assumes that the \textit{viability} domain, i.e., the set of initial conditions from which forward invariance of the safe set can be guaranteed, is known. In practice, it is not an easy task to compute the viability domain of a nonlinear control system. Optimization-based methods, such as Sum-of-Squares (SOS) techniques, have been used in the past to compute this domain (see \cite{wang2018permissive}). However, SOS-based approaches are only applicable to systems whose dynamics is given by polynomial functions, thus limiting their applications. Another method popularly used in the literature for computing the viability domain is Hamilton-Jacobi (HJ) based reachability analysis; see, e.g. \cite{choi2021robust}. However, such analysis is computationally expensive, particularly for higher-dimensional systems. We propose a novel sampling-based method to compute the viability domain for a general class of nonlinear control systems to overcome these limitations.

In this work, we consider a general class of nonlinear systems under actuator attacks and propose a method of computing a set of initial conditions and an input constraint set such that the system remains \textit{secure by design}. In particular, we consider actuator manipulation, where an attacker can assign arbitrary values to the input signals for a subset of actuators in a given bound. We consider the property of safety with respect to an unsafe set and propose sufficient conditions using sampling of the boundary of a set to verify whether the set is a \textit{viability domain} under attacks. Using these conditions, we propose a computationally tractable algorithm to compute the set of initial conditions and the input constraint set so that the safety of the system can be guaranteed under attacks. In effect, our proposed method results in a secure-by-design system that is resilient against actuator attacks.

In our previous work \cite{garg2021sampling}, we used a proactive scheme consisting of only designing a safe feedback law using CBF. One disadvantage of that approach is that the control is conservative because we assumed that the system could constantly be under attack. In contrast, this paper designs a reactive security mechanism that activates conservative control only after an attack is detected. We design a CBF-based attack detection mechanism and prove that it is sound, i.e., there are no false negatives in attack detection. Furthermore, we propose a hybrid control law to avoid Zeno behavior resulting from a naive switching in control policy upon attack detection. 

\subsection{Contributions}
We consider the safety property with respect to an unsafe set and propose an attack-detection mechanism based on the CBF condition for safety. We use an adaptive parameter based on how close the system is to violating the safety requirement and use this adaptive parameter in the attack detection to reduce conservatism. Based on the detection, we use a switching-based recovery from a \textit{nominal} feedback law (to be used when there is no attack) to a \textit{safe} feedback law when the system is under an adversarial attack. Then, we propose sufficient conditions using sampling of the appropriate set to verify whether the set is a \textit{viability domain} under attacks. Using these conditions, we propose a computationally tractable algorithm to compute the set of initial conditions and the input constraint set such that the system's safety can be guaranteed under attacks. In effect, our proposed method results in a secure-by-design system that is resilient against actuator attacks. Finally, we leverage these sets in a QP-based approach with provable feasibility for real-time online feedback synthesis. In contrast to the conference paper \cite{garg2021sampling,garg2022control}, this paper provides a detailed theoretical analysis and a complete proof of the analytical results. Furthermore, in this paper, we consider a more general class of dynamical systems modeled as differential inclusions, in contrast to the prior work where systems modeled under differential equations were studied. Finally, in the prior work \cite{garg2021sampling}, we used an off-the-shelf sampling algorithm based on the triangulation of spheres, while in this work, we propose a new sampling method that is computationally much more efficient than the triangulation-based methods.  The contributions of the paper are summarized below:
\begin{itemize}
    \item[1)] We present a novel attack detection mechanism using CBF conditions for safety. In the absence of knowledge of \textit{actual} system input under an attack, we utilize an approximation scheme and show that the attack-detection mechanism is sound, i.e., it does not generate any false negatives. While there is work on CBF-based safety of CPS under faults and attacks \cite{clark2020control,ramasubramanian2019linear}, to the best of the authors' knowledge, this is the first work utilizing CBF conditions for attack detection;
    \item[2)] Based on the zeroing-CBF condition \cite{ames2017control}, we propose an adaptation scheme to minimize the false-positive rate of the attack-detection mechanism. We propose a novel hybrid control law to keep the system safe under attacks with delays in detection and show that the resulting closed-loop system does not exhibit Zeno behavior;
    \item[3)] We present a novel, computationally efficient sampling technique for computing a viability domain that can be rendered forward invariant under adversarial attacks;
    \item[4)] Finally, we use a switching law for input assignment and a QP formulation for online feedback synthesis for both nominal and safe feedback. We illustrate the efficacy of the proposed method in a case study involving an attack on the motor of a quadrotor and show how the proposed framework can recover the quadrotor from an attack. 
\end{itemize}

\subsection{Organization and Notation}
The remainder of the paper is organized as follows. The formulation of the problem and the required preliminaries are presented in Section \ref{sec: Prob form}. The attack detection scheme is presented in Section \ref{sec: attack detect} while Section \ref{sec: attack recovery} presents a switched and a hybrid control scheme for attack recovery. Section \ref{sec: sampling} presents sampling-based methods for computing the necessary sets for attack recovery, and Section \ref{sec: QP control} presents a QP-based framework for online control synthesis. Section \ref{sec: numerical} presents numerical case studies, and the conclusions are presented in Section \ref{sec: conclusion}.

\noindent \textit{Notation}: Throughout the paper, $\mathbb N$ denotes the set of natural numbers ($0$ inclusive), $\mathbb R$ denotes the set of real numbers and $\mathbb R_+$ denotes the set of nonnegative real numbers. We use $|x|$ to denote the Euclidean norm of a vector $x\in \mathbb R^n$ and $|x|_{\mathcal A} = \inf_{y\in \mathcal A}|x-y|$, the distance of the point $x$ from the set $\mathcal A$. We use $\partial K$ to denote the boundary of a closed set $K\subset \mathbb R^n$ and $\textrm{int}(S)$ to denote its interior. The Lie derivative of a continuously differentiable function $h:\mathbb R^n\rightarrow\mathbb R$ along a vector field $f:\mathbb R^n\rightarrow\mathbb R^m$ at a point $x\in \mathbb R^n$ is denoted as $L_fh(x) \coloneqq \frac{\partial h}{\partial x}(x)f(x)$. The right limit of a function $z: \mathbb R_+ \to \mathbb R^n $ is given by 
$z^+ \coloneqq z(t^+) = \lim_{\tau \searrow t}z(\tau)$. The notation $\mathcal C^n$ is used to denote an $n-$times continuously differentiable function. {A continuous function $\alpha:\mathbb R_+\to\mathbb R_+$ is said to be a class-$\mathcal K$ function if it is strictly increasing and $\alpha(0) = 0$. The closure of an open set $\mathcal A$ is denoted as $\bar{\mathcal A}$.} 

\section{Problem Formulation}\label{sec: Prob form}


\subsection{System Model}
Consider a nonlinear control system $\mathcal S$ given as
\begin{align}\label{eq: actual system}
\mathcal S : \begin{cases}\dot x \in F(x,u) + d(t,x),\\
x\in \mathcal D, u\in \mathcal U,\end{cases}
\end{align}
where $F:\mathcal D\times\mathcal U\rightrightarrows \mathbb R^n$ is a known set-valued map with $\mathcal D\subset\mathbb R^n$ and $\mathcal U\subset \mathbb R^m$, $d:\mathbb R_+\times\mathbb R^n\rightarrow\mathbb R^n$ is unknown and represents the unmodeled dynamics, $x\in \mathcal D$ is the system state, and $u\in \mathcal U$ is the control input. For a given {Lebesgue measurable} input signal $u:\mathbb R_+\times\mathbb R^n\rightarrow\mathcal U$, a solution of $\mathcal S$ is a locally absolutely continuous function $x:\textrm{dom}~x\rightarrow\mathbb R^n$ satisfying $\dot x(t)\in F(x(t),u(t, x(t)))$ for almost all $t\in \textrm{dom}~x$, where $\textrm{dom}~x\subset\mathbb R_+$ is the domain of definition of $x$. A solution $x$ to $\mathcal S$ is complete if $\textrm{dom}~x$ is unbounded and is maximal if $\textrm{dom}~x$ cannot be extended.

\subsection{Attacker Model}
Similar to \cite{garg2021sampling}, in this paper, we consider attacks on the control input of the system. In particular, we consider an attack in which a subset of the components of the control input is compromised. Under such an attack, the system input takes the form:
\begin{align}\label{eq: attack model}
   u = (u_v, u_s),
\end{align}
where $u_v\in \mathcal U_v\subset \mathbb R^{m_v}$ represents the \textit{vulnerable} components of the control input that might be compromised or attacked, and $u_s\in \mathcal U_s\subset \mathbb R^{m_s}$ the \textit{secure} part that cannot be attacked, with $m_v + m_s = m$ and $\mathcal U\coloneqq \mathcal U_v\times \mathcal U_s$.
In this class of attack, we assume that we know which components of the control input are vulnerable. 

Similar attack models have been used in previous work; see, e.g., \cite{giraldo2020daria}, and can be implemented in practice by designing the dynamic range of the actuator to preserve its bounds. As discussed in \cite{pasqualetti2013attack}, various prototypical attacks, such as \textit{stealth attacks}, \textit{replay attacks}, and \textit{false-data injection attacks}, can be captured by the attack model in \eqref{eq: attack model}. In addition to representing a real-world scenario in which system actuators have physical limits, restricting the vulnerable control input $u_v$ in the set $\mathcal U_v$ has the following advantages:
\begin{itemize}
    \item[1)] It restricts how much an attacker can change the nominal operation of the system \cite{kafash2018constraining}, and can be implemented physically, so that an attacker cannot bypass it.
    \item[2)] It can be utilized to design a detection mechanism, e.g., if $u_v\notin \mathcal U_v$, a flag can be raised, signifying that the system is under attack. Schemes that raise a threshold-based flag are commonly used as detection mechanisms \cite{renganathan2020distributionally}. 
\end{itemize}

\begin{figure}[t]
	\centering
	\includegraphics[width=1\columnwidth,clip]{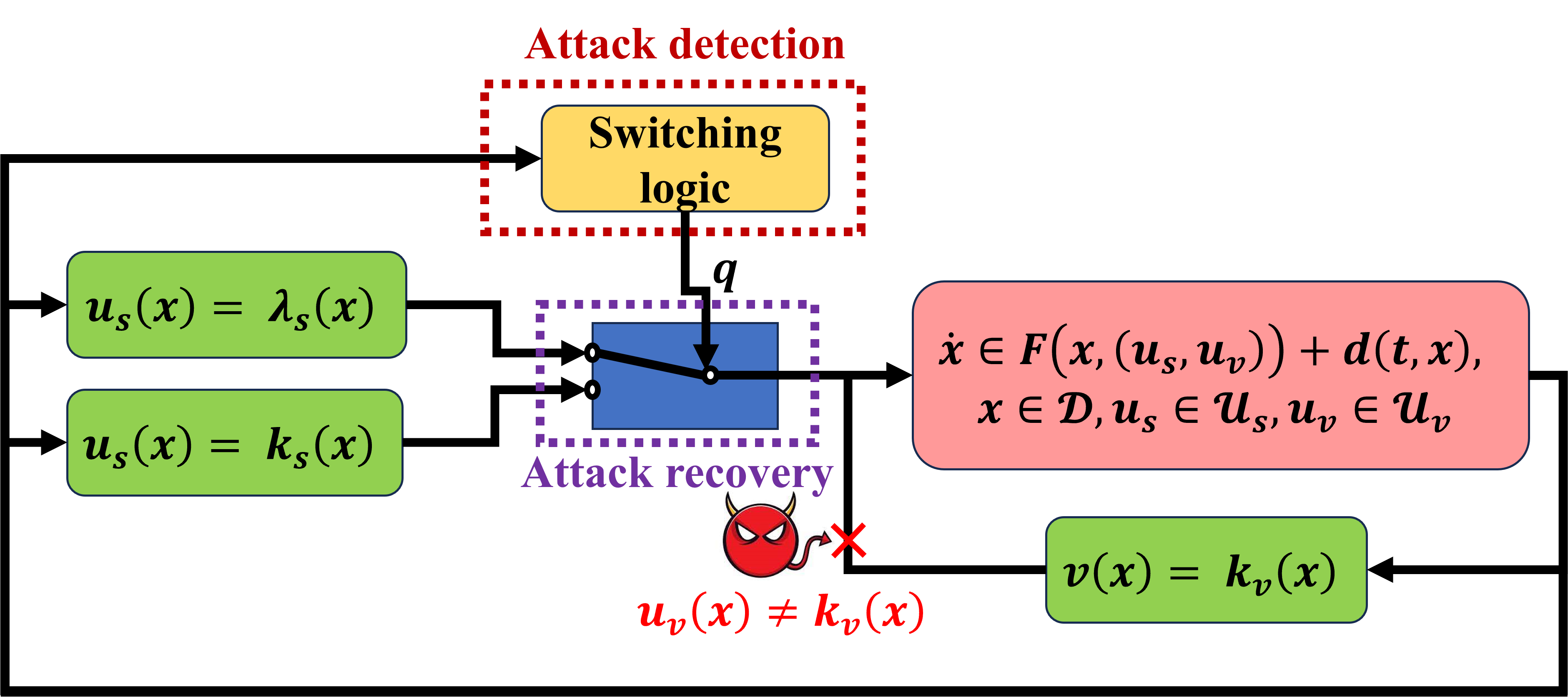}
	\caption{Overview of the proposed attack-detection-based approach for attack recovery.}
	\label{fig:overview}
\end{figure}

Under this attack model, the input to the system takes the form:
\begin{align}\label{eq: attack input}
   u(t,x) = \begin{cases} (\lambda_v(x),\lambda_s(x))  & \textrm{if}\; t\notin  \mathcal T_a \\
   (u_a(t),k_s(x)) & \textrm{if} \; t\in \mathcal T_a \end{cases},
\end{align}
where $u_a:\mathbb R_+\rightarrow \mathcal U_v$ is the attack signal on the input $u_v$, $k_s:\mathbb R_+\times\mathbb  R^n\rightarrow\mathbb R^{m_s}$ is a \textit{safe} feedback law for the input $u_s$, which is to be designed and used when the system is under attack, and the pair $\lambda_v:\mathbb R^n\rightarrow\mathcal U_v, \lambda_s:\mathbb R^n\rightarrow \mathcal U_s$ defines the \textit{nominal} feedback law $\lambda = (\lambda_v, \lambda_s)$, to be designed and used when there is no attack (see Figure \ref{fig:overview}). The set $\mathcal T_a\subset\mathbb R_+$ is the set of time intervals when an attack is launched on the system input. In particular, for each $i\geq 1$, let $[t^i_1, t^i_2)$ with $t^i_2\geq t^i_1$ denote the interval of time when the attack is launched for the $i-$th time where $t^1_1\geq 0$, so that $\mathcal T_a \coloneqq \bigcup\limits_{i\geq 0}[t^i_1, t^i_2)$. Define 
\begin{align}\label{eq: max min T attack}
    \overline T & \coloneqq \max_{i\geq 1}\{t^i_2-t^i_1\},\\
    T_{na} & \coloneqq \min_{i\geq 2}\{t^i_1-t^{(i-1)}_2\},
\end{align}
as the maximum length of the attack and the minimum length of the interval without an attack on the system input, respectively. In this work, we assume that the set $\mathcal T_a$ is unknown, and only the maximum period of attack, $\overline T$, and minimum period without an attack, $T_{na}$, are known. We make the following assumption about $\mathcal S$.

\begin{Assumption}\label{assum: d bound}
The map $(t, x) \mapsto F(x, u(t, x)) + d(t,x)$ is lower semicontinuous, has nonempty, closed, and convex values for all $(t, x)\in \mathbb R_{\geq 0}\times \mathcal D$. 
Furthermore, there exists a known $\delta>0$ such that $|d(t, x)|\leq \delta$ for all $t\geq 0$ and all $x\in \mathcal D$.
\end{Assumption}

Under Assumption \ref{assum: d bound}, from \cite[Ch. 2, Theorem 1]{aubin2012differential}, it holds that at least one solution of \eqref{eq: actual system} is continuously differentiable.\footnote{Note that there are stronger assumptions required on $F$ for uniqueness of solutions. In this work, we do not make such assumptions and allow $\mathcal S$ to have nonunique solutions.}
Now, we present the control design problem studied in this paper. Consider a nonempty, compact set $K\subset\mathbb R^n$, referred to as a safe set, to be rendered forward invariant. 



\begin{Problem}\label{Problem 1}
Given the system in \eqref{eq: actual system} with unmodeled dynamics $d$ that satisfies Assumption \ref{assum: d bound}, a set $K\subset \mathcal D$, and the attack model in \eqref{eq: attack model}, design an attack-detection mechanism to raise a flag that the system is under attack
and apply a safe input assignment policy
such that, for a set of initial conditions $X_0\subseteq K$ and attack signals $u_a:\mathbb R_+\rightarrow \mathcal U_v$, for all $t\in \textnormal{dom}~x$ and for each $x(0)\in X_0$, each closed-loop solution $x:\textnormal{dom}~x\rightarrow\mathbb R^n$ of \eqref{eq: actual system} resulting from applying the designed input policy
satisfies $x(t) \in K$.
\end{Problem}

Note that for the safety requirement as imposed in Problem~\ref{Problem 1}, an attack is adversarial only if it can push the system trajectories out of the set $K$ for any input assignment, as defined below.

\begin{Definition}\label{def: advers attack}
An attack signal $u_a:\mathbb R_+\rightarrow\mathcal U_v$ is \textnormal{adversarial} if there exist $x_0\in K$ and a finite $t\in \textnormal{dom}~x$ such that for any $\kappa:\mathbb R_+\times\mathbb R^n\rightarrow\mathcal U_s$, there exists a solution $x:\textnormal{dom}~x\rightarrow\mathbb R^n$ of \eqref{eq: actual system} resulting from applying $u = (u_a, \kappa)$ with $x(0) = x_0$ such that $x(t)\notin K$ for some $t\in \textnormal{dom}~x$. 
\end{Definition}

According to the above definition, it is possible that there is an attack on the system but the system does not violate the safety requirement. {We are not concerned about such \textit{non-adversarial} attack signals in this work.} We use this observation to focus our detection mechanism only on adversarial attacks that can potentially push the system out of the safe set.  

\subsection{Mathematical Preliminaries}

Following \cite{chai2018forward}, we define the notion of forward pre-invariance and forward invariance of a set $K\subset \mathbb R^n$ for $\mathcal S$. 

\begin{Definition}
   A set $K\subset\mathbb R^n$ is said to be \textnormal{forward pre-invariant} for system \eqref{eq: actual system} if, for each $x_0\in K$, each maximal solution $x$ starting at $x(0) = x_0$ satisfies $x(t)\in K$ for all $t\in \textnormal{dom}~x$. If, in addition, each maximal solution is complete, then the set $K$ is said to be \textnormal{forward invariant}.
\end{Definition}
Next, building from \cite{maghenem2021sufficient,chai2020forward}, we formulate a sufficient condition for guaranteeing forward pre-invariance of a set without an attack. 

\begin{Lemma}\label{lemma: nec suff safety}
Given a continuously differentiable function $B:\mathbb R^n\rightarrow\mathbb R$, the set $K = \{x\; |\; B(x)\leq 0\}$ is forward pre-invariant for $\mathcal S$ under $d$ satisfying Assumption \ref{assum: d bound} with $\delta>0$ if there exists a neighborhood $U (\partial K)$ of the boundary $\partial K$ such that
\begin{equation}\label{eq: safety cond}
    \inf_{u\in \mathcal U}\sup_{\zeta\in F(x,u)}L_{\zeta}B(x,u)\leq -l_B\delta \quad \forall x\in (U(\partial K)\setminus K),
\end{equation}
where $l_B$ is the Lipschitz constant of the function $B$.
\end{Lemma}

We also review a solution-based safety condition, in which the CBF is evaluated along a closed-loop solution of \eqref{eq: actual system}. 

\begin{Lemma}\label{lemma: suff cond sol based}
{Given a continuously differentiable function $B:\mathbb R^n\rightarrow\mathbb R$, under Assumption \ref{assum: d bound}, consider a $\mathcal C^1$ closed-loop solution $x:\textnormal{dom}~x\rightarrow\mathbb R^n$ with $x(0)\in K= \{x\; |\; B(x)\leq 0\}$ of $\mathcal S$ resulting from using a feedback $k:\mathbb R_+\times \mathbb R^n\rightarrow\mathcal U$ under $d$ satisfying Assumption \ref{assum: d bound}. The set $K$ is forward pre-invariant for $\mathcal S$ if}
\begin{equation}\label{eq: safety cond sol based}
    \frac{d}{dt}B(x(t))\leq 0 \quad \forall t\in \{t\in \textnormal{dom}~x\; |\; B(x(t)) = 0\}.
\end{equation}
\end{Lemma}

Finally, in this work, we use second-order Taylor's expansion of a $\mathcal C^1$ function, which requires the following notion of generalized Hessian.

\begin{Definition}{\hspace{-1pt}\cite[Def. 1.1]{cominetti1990generalized}}
    The generalized second-order gradient of a function $\phi:\mathbb R^n\rightarrow\mathbb R$ at $x\in \mathbb R^n$ in the direction $(u, v)\in \mathbb R^n\times \mathbb R^n$ is given as
    \begin{align}
        & {\phi^\infty(x, (u, v))} \nonumber \\
        = & \limsup_{
                                      y\to x \atop
                                      t, s \to 0
                                  }\frac{{\phi(y + su + tv) - \phi(y + su)- \phi(y+ tv) + \phi(y)}}{st}
    \end{align}
    and the generalized Hessian of $\phi$ at $x$ {in the direction $u\in \mathbb R^n$} is given as
    \begin{align}
        \partial^2 \phi(x,u) = \{z\in \mathbb R^n \; | \; z ^\top v \leq \phi ^\infty (x, (u, v))~ \forall v\in \mathbb R^n\}.
    \end{align}
\end{Definition}

The following lemma reviews the second-order Taylor's expansion of functions that are not $\mathcal C^2$ (adopted from \cite[Proposition 4.1]{cominetti1990generalized}) using the generalized Hessian. 
\begin{Lemma}\label{lemma: taylor}
    Given a continuously differentiable function $\psi:\textnormal{dom}~\psi\rightarrow\mathbb R$, {where $\textnormal{dom}~\psi\subset\mathbb R_+$}, with lower semicontinuous generalized Hessian $\partial^2 \psi$, for each $t, \mathcal T>0$ with {$t, t-\mathcal T\in \textnormal{dom}~\psi$,} there exists $\tau\in [0, \mathcal T]$ such that
    \begin{align}
        \psi(t) - \psi(t-\mathcal T) - \mathcal T \dot \psi(t) \in \frac{\mathcal T^2}{2}\overline{\partial^2\psi(t - \tau, \mathcal T)}.
    \end{align}
    If, in addition, for each $\mathcal T > 0$, there exists $\eta>0$ such that $\overline{\partial^2\psi(t, \mathcal T)}\leq \eta$ for all $t\in \textnormal{dom}~\psi$, then the following holds:
    \begin{align}
        \left|\frac{\psi(t)-\psi(t-\mathcal T)}{\mathcal T} - \dot \psi(t)\right| \leq \frac{\mathcal T}{2} \eta
    \end{align}
    for all $t, t-\mathcal T\in \textnormal{dom}~\psi$.
\end{Lemma}

We briefly review the notion of hybrid systems and its solutions as these concepts become useful later in the paper. A hybrid system is given as \cite{goebel2012hybrid}:
\begin{align}
    \mathcal H : \begin{cases}
        \dot z = f(z) & z\in C, \\
        z^+ = g(z) & z\in D,
    \end{cases}
\end{align}
with state variable $z\in \mathbb R^n$, flow map $f:\mathbb R^n\rightarrow\mathbb R^n$, jump map $g:\mathbb R^n\rightarrow\mathbb R^n$, flow set $C\subset\mathbb R^n$, and jump set $D\subset \mathbb R^n$. A solution to $\mathcal H$ is defined on the hybrid time domain $\textnormal{dom}~z\subset \mathbb R_+\times \mathbb N$, which parameterized the solution by continuous time $t\in \mathbb R_+$ and discrete time $j\in \mathbb N$. A hybrid time domain is a subset of $\mathbb R_+\times \mathbb N$ such that for every $(T, J)\in \textnormal{dom} z$, there exists a sequence $\{t_j\}_{j = 0}^{J+1}$ such that $t_0 = 0$, $t_{j+1}\geq t_j$ for each $j\in \{0, 1, \dots, J\}$, and $\textnormal{dom} ~z\cap ([0, T]\times \{0, 1, \dots, J\}) = \cup_{j = 0}^J[t_j, t_{j+1}], j)$ (see, e.g., \cite{goebel2012hybrid}). A solution $z$ to $\mathcal H$ is said to be \textit{complete} if $\textnormal{dom}~z$ is unbounded and is said to be \textit{Zeno} if it is complete, and the $t$ component of $\textnormal{dom}~ z$ is bounded. A solution $z$ is said to be \textit{maximal} if there does not exist a solution $y$ to $\mathcal H$ such that $\textnormal{dom} ~ z\subset \textnormal{dom}~y$. 

\section{Attack Detection}\label{sec: attack detect}
\subsection{CBF-based Detection}
In this section, we present a method of detecting whether the system \eqref{eq: actual system} is under attack using the barrier function condition \eqref{eq: safety cond}. In particular, {if the inequality \eqref{eq: safety cond} is violated on the boundary of the safe set, then an adversarial attack is flagged}. In contrast to using the value of the barrier function $B$, we use the value of its time derivative as it includes the system dynamics. 
Hence, the value of the time derivative of the function $B$ is a better indicator of whether the given system will violate the given safety constraint compared to the value of the function $B$ itself, which does not capture the system information. Note that if an attack signal $u_a$ is adversarial as per Definition \ref{def: advers attack}, then it holds that there exists a finite time $t\geq 0$ such that $x(t)\in (U(\partial K)\setminus K)$ and $\inf_{u_s\in \mathcal U_s}\sup_{\zeta\in F(x,(u_a, u_s))}L_{\zeta}B(x,u)>-l_B\delta$. Using this, a detection mechanism can be devised to flag that the system input is under attack. When the input $u$ to the system is known at time $t$ when $x(t)\in \partial K$, we propose an attack detection mechanism that checks the value of $\sup_{\zeta\in F(x,u)} L_\zeta B(x(t),u)$ to flag an attack. 

However, in the presence of an unknown attack, it is not possible to know the actual input $u$ to the system. Thus, it is not possible to use the evaluation of $L_\zeta B$ to flag an attack. 
To this end, we can obtain second-order Taylor expansion of the function $B$, evaluated along a closed-loop system trajectory $x:\mathbb R_+\rightarrow\mathbb R^n$ in order to obtain an approximation of the time derivative $\dot B(x(t))$ when $u$ is unknown. 

Let $\tau>0$ be the sampling-time period for the first-order approximation of $\dot B$ using consecutive measurements of the function $B$. 
Under the assumption that the function $B$ is continuously differentiable, it follows that for any continuously differentiable solution $x:\textnormal{dom}~x\rightarrow\mathbb R^n$ of \eqref{eq: actual system}, the composite function $B \circ x$ is continuously differentiable on $\textnormal{dom}~x$. 
Define $e_B:\textnormal{dom}~x\rightarrow\mathbb R$ as
\begin{align*}
    e_B(t) \coloneqq \left|\frac{d}{dt}B(x(t)) - \frac{B(x(t))-B(x(t-\tau))}{\tau}\right|,
\end{align*}
which is the error between the derivative of the function $B$ and its first-order approximation. 

In order to obtain a bound on $e_B$, we make the following assumption.
\begin{Assumption}\label{assum: sec der bound}
For each continuously differentiable solution $x:\textnormal{dom}~x\rightarrow\mathbb R^n$ of \eqref{eq: actual system} under an input $u:\mathbb R_+\rightarrow\mathcal U$ with $x(0)\in K$, there exist $\tau, \eta>0$ such that
\begin{align}
    \left|{\partial^2 B(x(t), x(t)-x(t-\tau))}\right|\leq \eta
\end{align}
for all $t\in \textnormal{dom}~x$. 
\end{Assumption}

{
\begin{Remark}
Assumption \ref{assum: sec der bound} aids Lemma \ref{lemma: taylor} by assuming the required bound of the generalized Hessian of the map $B\circ x$. Per discussion in \cite{cominetti1990generalized}, the map $B\circ x$ satisfies conditions of Lemma \ref{lemma: taylor} if it is of class $\mathcal C^{1,1}$, i.e., it is continuously differentiable with a Lipschitz continuous gradient. We leave further relaxation of this regularity condition as future work and refer the interested reader to the related literature \cite{clarke1995proximal,rockafellar1981favorable,poliquin1996generalized}.
\end{Remark}
}

Let $x:\mathbb R_+\rightarrow\mathbb R^n$ be the solution of \eqref{eq: actual system} resulting from applying the input $u:\mathbb R_+\rightarrow \mathcal U$. Under Assumptions \ref{assum: d bound} and \ref{assum: sec der bound}, using Lemma \ref{lemma: taylor}, it holds that 
\begin{align*}
    \left|\frac{B(x(t)-B(x(t-\tau))}{\tau} - \dot B(x(t))\right| \leq  \eta\frac{\tau}{2},
\end{align*}
{for each $t\geq \tau$ and $t\in \textnormal{dom}~x$}. 
For the sake of brevity, define 
\begin{align}\label{eq: dot B hat}
    \hat{\dot B}(x(t),\tau) \coloneqq \frac{B(x(t))-B(x(t-\tau))}{\tau},
\end{align}
so that we have
\begin{align*}
  e_B(t) = |\dot B(x(t))-\hat{\dot B}(x(t),\tau)| \leq \frac{\eta \tau}{2}. 
\end{align*}
Thus, it holds that $e_B(t) \leq \frac{\eta \tau}{2}$. Using the bound on $e_B$, we obtain that for each $t\geq 0$, the following holds:
\begin{align}\label{eq: H B bound}
    \hat{\dot B}(x(t),\tau)-\frac{\eta\tau}{2}\leq \dot B(x(t))\leq \hat{\dot B}(x(t),\tau)+\frac{\eta\tau}{2}.
\end{align}
Then, with $t, \tau\geq 0$, it follows that
\begin{align*}
 \hat{\dot B}(x(t),\tau)+\frac{\eta\tau}{2}\leq 0   \implies \dot B(x(t))\leq 0.
\end{align*}
With the above construction, we propose the following attack detection mechanism:
\begin{itemize}
    \item[1)] Given $\tau>0$ and $\bar t\geq 0$ such that $x(\bar t)\in \partial K$, evaluate $\hat {\dot B}(x(\bar t), \tau)$. 
    \item[2)] If $\hat {\dot B}(x(\bar t), \tau)>-\frac{\eta\tau}{2}$, raise a flag that the system is under attack.
\end{itemize}

More concisely, we define the time when a flag for an attack is raised as
\begin{align}\label{eq: t_d estimate}
    \hat t_d = \inf\left\{t\; \Big |\; \hat{\dot B}(x(t),\tau)>-\frac{\eta\tau}{2}, x(t)\in \partial K\right\},
\end{align}
where $\eta$ is the bound on the generalized Hessian $\partial ^2B$ and $\tau>0$.
We have the following result stating that the attack detection mechanism in \eqref{eq: t_d estimate} detects the attack before the system trajectories leave the safe set. 

\begin{Lemma}\label{Lemma td hat td relation}
Given a twice continuously differentiable function $B$, system \eqref{eq: actual system} with $d$ satisfying Assumption \ref{assum: d bound}, a continuously differentiable map $F$, and an adversarial attack starting at $t = t^i_1$, let $T\geq t^i_1$ be defined as 
\begin{align}
    T = \inf\Big\{t\geq t^i_1\; |\; \dot B(x(t)) & > 0, x(t)\in \partial K\Big\},
\end{align}
where $x:\textnormal{dom}~x\rightarrow\mathbb R^n$ is any solution of \eqref{eq: actual system} resulting from applying the input $u:\textnormal{dom}~x\rightarrow \mathcal U$ with $x(0)\in K$ and $\eta$ is as per Assumption \ref{assum: sec der bound}. Then, for each $\tau\geq 0$, it holds that $\hat t_d\leq T$, where $\hat t_d$ is given in \eqref{eq: t_d estimate}.
\end{Lemma}

\begin{proof}
{Under the smoothness assumptions on $F,B$, the map $B\circ x$ satisfies the conditions of Lemma \ref{lemma: taylor}, which enables the existence of $\eta$ per Assumption \ref{assum: sec der bound}.} If $\hat t_d>T$, it holds that there exists $t\in (T, \hat t_d)$ such that $\dot B(x(t))>0$ and $\hat{\dot B}(x(t),\tau)\leq-\frac{\eta\tau}{2}$. Using this along with the second inequality in  \eqref{eq: H B bound} at time instant $t$, we obtain that
\begin{align*}
    0< \dot B(x(t)) \leq \hat{\dot B}(x(t),\tau)+\frac{\eta \tau}{2}\leq 0,
\end{align*}
which is a contradiction and hence, $\hat t_d\leq T$. 
\end{proof}

Lemma \ref{Lemma td hat td relation} implies that the attack-detection mechanism in \eqref{eq: t_d estimate} raises an alert on or before the system trajectories reach the boundary of the set $\partial K$ under an attack. In other words, while the detection-mechanism \eqref{eq: t_d estimate} can have false positives (i.e., raise an alert when there is no attack), it will never have a false negative (i.e., it will not miss any attack).

\subsection{Adaptive Scheme for ZCBF-based Attack Detection}
One of the limitations of using the inequality \eqref{eq: safety cond} at the boundary of the safe set $K$ for detecting an attack is that it is not robust due to the following two reasons: (i) any small measurement uncertainty or disturbance can lead to violation of safety, and (ii) any nonzero delay in responding to the attack can lead to violation of safety. Assume that the set $K$ is compact and let $K_c \coloneqq \{x\; |\; B(x)\leq -c\}$ be a sublevel set of the function $B$ for a given $c\geq 0$. Using this, one method to make the detection method robust is to check the inequality at the boundary of the set $K_c$ for some $c>0$. Define $c_M\in \mathbb R$ as
\begin{align}\label{eq: cM}
c_M \coloneqq -\min\limits_{x\in K}B(x),
\end{align}
so that the set $K_c$ is nonempty for all $c\in [0, c_M)$.\footnote{Compactness of the set $K$ guarantees existence of $c_M\in \mathbb R_+$.} Define 
\begin{align}\label{eq: func H}
    H(x) \coloneqq \inf_{u\in \mathcal U}\sup_{\zeta\in F(x,u)}L_{\zeta}B(x,u) + l_B\delta.
\end{align}
Now, since it is possible to allow the function $H$ to take positive values in the interior of the safe set $K$, we use the inequality $H(x)\leq \gamma$ for some $\gamma>0$ instead of $H(x)\leq 0$, to detect attacks. Note that a constant $\gamma>0$ might lead to false positives if $\gamma$ is too small or false negatives if $\gamma$ is too large. To this end, we make the following assumption when the system is not under attack. 

\begin{Assumption}\label{assum: dot B bar c}
There exist $\bar c\in (0, c_M)$, $\bar \delta\in \mathbb R$ and a continuous feedback $\bar k:\mathbb R^n\rightarrow\mathcal U$ such that the following inequality holds for all $x\in K\setminus\textnormal{int}(K_{\bar c})$:
\begin{align}\label{eq: dot H bar c}
    \inf_{u\in \mathcal U}\sup_{\zeta\in F(x,u)}L_{\zeta}B(x,u)\leq -\bar \delta B(x) - l_B \delta,
\end{align}
where $K_{\bar c} = \{x\; |\; B(x) \leq -\bar c\}$, $\delta>0$ is the bound on the disturbance $d$ from Assumption \ref{assum: d bound}, and $l_B>0$ is the Lipschitz constant of the function $B$. 
\end{Assumption}
Similar assumptions have been made in the literature on safety using ZCBFs (see, e.g., \cite{ames2017control}). Note that under Assumption \ref{assum: dot B bar c}, using the comparison lemma, it can be shown that 
\begin{align}\label{eq: B comp lemma }
\dot B(x(t)) &  \leq -\bar \delta B(x(t)) 
\implies 
 B(x(t)) \leq  B(x(\bar t))e^{-\bar\delta(t-\bar t)},
\end{align}
for all $t\geq \bar t$, where $\bar t = \inf\{ t\; |\; x(t)\in \partial K_{\bar c}\}$ and $x:\mathbb R_+\rightarrow\mathbb R^n$ is the solution of \eqref{eq: actual system} resulting from applying the feedback $\bar k$. Now, we design an adaptive scheme for the parameter $\gamma$. Let $\gamma:\mathbb R_+\rightarrow\mathbb R_+$ be an adaptive parameter whose adaptation law is given as 
\begin{align}\label{eq: gamma delta c}
    \gamma(t) = -\bar \delta B(x(t)),
\end{align}
for $t\geq \bar t$, where $\delta>0$ is as defined in Assumption \ref{assum: d bound} and $\bar\delta$ is as defined in Assumption \ref{assum: dot B bar c}. Note that under Assumption \ref{assum: dot B bar c}, there exists a feedback law $\bar u:\mathbb R^n\rightarrow\mathcal U$ such that $\dot B(x(t))\leq \gamma(t)$ for all $t\geq \bar t$, where $x$ is the resulting trajectory under $\bar u$. Using this observation, we propose a new attack-detection mechanism that raises a flag for the $i-$th time at $t = \hat t_d^i$, where 
\begin{align}\label{eq: t_d estimate adapt}
    \hat t_d^i = \inf\Big\{t \geq \max\left\{\bar t, \hat t_d^{(i-1)}\right\}\; \Big |\; & \hat{\dot B}(x(t),\tau) > \gamma(t)-\frac{\eta\tau}{2},\nonumber \\
    & x(t)\in K\setminus\textnormal{int}(K_{\bar c})\Big\},
\end{align}
where $\eta$ is the bound on the generalized Hessian $\partial^2 B$, $\gamma$ is as defined in \eqref{eq: gamma delta c}, $\hat t_d^0 = -\overline T$, and $\tau>0$.
\begin{Remark}
Under an attack, the proposed detection mechanism allows the system to get closer to the boundary of the safe set as long as the rate at which the system approaches the boundary (dictated by the time derivative function $\dot B$) is bounded according to Assumption \ref{assum: dot B bar c}. Also, it should be noted that the proposed attack detection mechanism focuses on detecting only \textnormal{{adversarial}} attacks (see Definition \ref{def: advers attack}), and not every attack. That is, if there is an attack on the system that cannot push the state out of the safe set, the proposed detection mechanism will not detect it. Thus, the proposed mechanism will have false positives as well as false negatives (for non-adversarial attacks). 
\end{Remark}


\section{Attack Recovery}\label{sec: attack recovery}
\subsection{Switching Control Law for Recovery}
In this section, we present a switching-based control assignment to recover from an adversarial attack based on the detection mechanism \eqref{eq: t_d estimate adapt} from the previous section. To this end, we make the following assumption. 

\begin{Assumption}\label{assum: Sc set}
Given the compact set $K = \{x \; |\; B(x) \leq 0\}$ and system $\mathcal S$ in \eqref{eq: actual system}, there exists $\bar c\in (0, c_M)$, where $c_M$ is as given in \eqref{eq: cM}, such that the following hold for each $x\in  K\setminus \textnormal{int} (K_{\bar c})$:
\begin{align}\label{eq: Assum Sc attack}
    \inf_{u_s\in \mathcal U_s}\sup_{u_a\in \mathcal U_a}\sup_{\zeta\in F(x,(u_a,u_s))}L_{\zeta} B(x,(u_a, u_s))  & \leq -l_B\delta, 
\end{align}
where $K_{\bar c} = \{x\; |\; B(x) \leq -\bar c\}$, $\delta>0$ is as defined in Assumption \ref{assum: d bound}, and $l_B>0$ is the Lipschitz constant of the function $B$. 
\end{Assumption}

The above assumption implies that 
the set $K_{c}$ can be rendered forward invariant under any attack $u_a\in \mathcal U_a$ for any $c\in (0, \bar c]$.
Based on the detection scheme in the previous section, we propose a switching-based control assignment for attack recovery. 
Consider a time-interval $[t^{(i-1)}_2, t^i_1)$ over which the system input is not under attack and suppose it is under an attack over $[t^i_1, t^i_2)$. Define $\mathcal T_d \coloneqq \bigcup_{j = 0}^{\infty} [\hat t_d^j, \hat t_d^j+\overline T)$ as the set of time intervals when an attack is flagged, where $\hat t_d^j$ is the time when the attack is flagged for the $j-$th time, $j\geq 0$, with $\hat t_d^0 = -\overline T$. Since $\mathcal T_a$ in \eqref{eq: attack input} is unknown, the system input is defined as
\begin{align}\label{eq: switch input}
     u(t,x) & = \begin{cases}(\lambda_v(x), \lambda_s(x)) & \textrm{if}\quad  t\notin \mathcal T_a \bigcup \mathcal T_d,\\
    (u_a(t), \lambda_s(x)) & \textrm{if} \quad t\in \mathcal T_a\setminus \mathcal T_d, \\
    (u_a(t), k_s(x)) & \textrm{if} \quad t\in \mathcal T_a \bigcap \mathcal T_d, \\
    (\lambda_v(x), k_s(x)) & \textrm{if} \quad t\in \mathcal T_d \setminus \mathcal T_a. \\
    \end{cases}
\end{align}
{Recall that $(\lambda_v, \lambda_s)$ constitute the nominal feedback laws, $u_a$ the attack signal and $k_s$ the recovery feedback law. Note that the secure inputs switch from nominal feedback $\lambda_s$ to $k_s$ upon detection of the attack (i.e., $t\in \mathcal T_d$), while the vulnerable inputs switch from $\lambda_v$ to $u_a$ when the attack begins (i.e., $t\in \mathcal T_a$).}

We have the following result showing the existence of nominal and safe feedback laws for \eqref{eq: switch input} that can recover the system from an attack.

\begin{Theorem}\label{thm: suff safety under attack}
Given system \eqref{eq: actual system} with $F\in \mathcal C^1$, $B\in \mathcal C^2$ and the attack model \eqref{eq: attack model}, suppose that Assumption \ref{assum: d bound} holds, and Assumptions \ref{assum: dot B bar c}-\ref{assum: Sc set} hold for some $\bar c\in (0, c_M)$. 
Then, there exist feedback laws $\lambda_v:\mathbb R^n\rightarrow\mathcal U_v$, $\lambda_s:\mathbb R^n\rightarrow\mathcal U_s$ and $k_s:\mathbb R^n\rightarrow\mathcal U_s$ such that under the effect of the input $u$ in \eqref{eq: switch input} with $\hat t_d^j$ is defined in \eqref{eq: t_d estimate}, $\textnormal{dom}~x = \mathbb R_+$ and 
the system trajectories of \eqref{eq: actual system} resulting from applying \eqref{eq: switch input} satisfy $x(t)\in K$ for all $t\geq 0$ and for all $x(0)\in X_0 = \textnormal{int}(K)$. 
\end{Theorem}
\begin{proof}
{Let $x:\textnormal{dom}\to \mathbb R^n$ be a solution of \eqref{eq: actual system} under the input \eqref{eq: switch input} with initial condition $x(0)\in \textnormal{int}(K)$} and consider the four cases: $t\in \mathcal T_a\setminus \mathcal T_d$, $t\in \mathcal T_a\cap \mathcal T_d$, $t\in \mathcal T_d\setminus\mathcal T_a$ and $t\notin \mathcal T_a \bigcup \mathcal T_d$.

Case 1: $t\in \mathcal T_a\setminus \mathcal T_d$. Since $t\notin \mathcal T_d$, from the definition of $\hat t_d$ in \eqref{eq: t_d estimate adapt}, it holds that either $x(t)\in \textnormal{int}(K_{\bar c})$ or $x(t)\in K\setminus\textnormal{int}(K_{\bar c})$ and $H(x)\leq 0$ where $H$ is defined in \eqref{eq: func H}. Thus, it holds that $x(t)\in \textnormal{int}(K)$ for all $t\in \mathcal T_a\setminus \mathcal T_d$.

Case 2: $t\in \mathcal T_a\cap \mathcal T_d$. Per Assumption \ref{assum: Sc set}, {it holds that there exists a feedback law $k_s:\mathbb R^n\rightarrow\mathcal U_s$, given as
\begin{align*}
    k_s(x) = \arginf\limits_{u_s\in \mathcal U_s}\sup_{u_a\in \mathcal U_a}\sup_{\zeta\in F(x,(u_a,u_s))}L_{\zeta}B(x, (u_a, u_s)),
\end{align*}}\textnormal
such that the set $K_{\hat c}$ is forward invariant for \eqref{eq: actual system} with $u(t,x) = (u_a(t), k_s(x))$ for any $u_a:\mathbb R_+\rightarrow\mathcal U_v$. Thus, it holds that $x(t)\in \textnormal{int}(K\setminus\textnormal{int}(K_{\bar c}))\subset\textnormal{int}(K)$ for all $t\in \mathcal T_a\cap \mathcal T_d$.

Case 3: $t\in \mathcal T_d\setminus \mathcal T_a$. Since in this time interval, there is no attack, {the feedback law $k_s$ can be defined as
\begin{align*}
    k_s(x) = \arginf\limits_{u_s\in \mathcal U_s}\sup_{\zeta\in F(x,(\lambda_v(x),u_s))}L_{\zeta}B(x, (\lambda_v(x), u_s)).
\end{align*}}\textnormal
Thus, $x(t)\in \textnormal{int}(K\setminus\textnormal{int}(K_{\bar c})\subset\textnormal{int}(K)$ for all $t\in \mathcal T_d\setminus \mathcal T_a$.

Case 4: $t\notin \mathcal T_a\bigcup \mathcal T_d$. In this case, per Assumption \ref{assum: dot B bar c}, there exists feedback laws {$\lambda_v, \lambda_s$ given as $(\lambda_v(x), \lambda_s(x)) = \lambda(s)$ where
\begin{align*}
    \lambda(x) = \arginf\limits_{u\in \mathcal U}\sup_{\zeta\in F(x, u)}L_\zeta B(x, u).
\end{align*}}\textnormal
Hence, the set $K$ is forward invariant for \eqref{eq: actual system} under $u = \lambda (x)$. 

Thus, it holds that $x(t)\in K$ for all $t\in \textnormal{dom}~x$ and $x(0)\in \textnormal{int}(K)$. Since the set {$K$ is assumed to be compact}, it follows from \cite[Ch. 2, Theorem 1]{aubin2012differential} that $\textnormal{dom}~x = \mathbb R_+$, and thus, the set $K$ is forward invariant for \eqref{eq: actual system}.
\end{proof}

In essence, Theorem \ref{thm: suff safety under attack} provides sufficient conditions for the existence of a control algorithm such that Problem \ref{Problem 1} can be solved. 

\begin{figure}[t]
	\centering
	\includegraphics[width=0.8\columnwidth,clip]{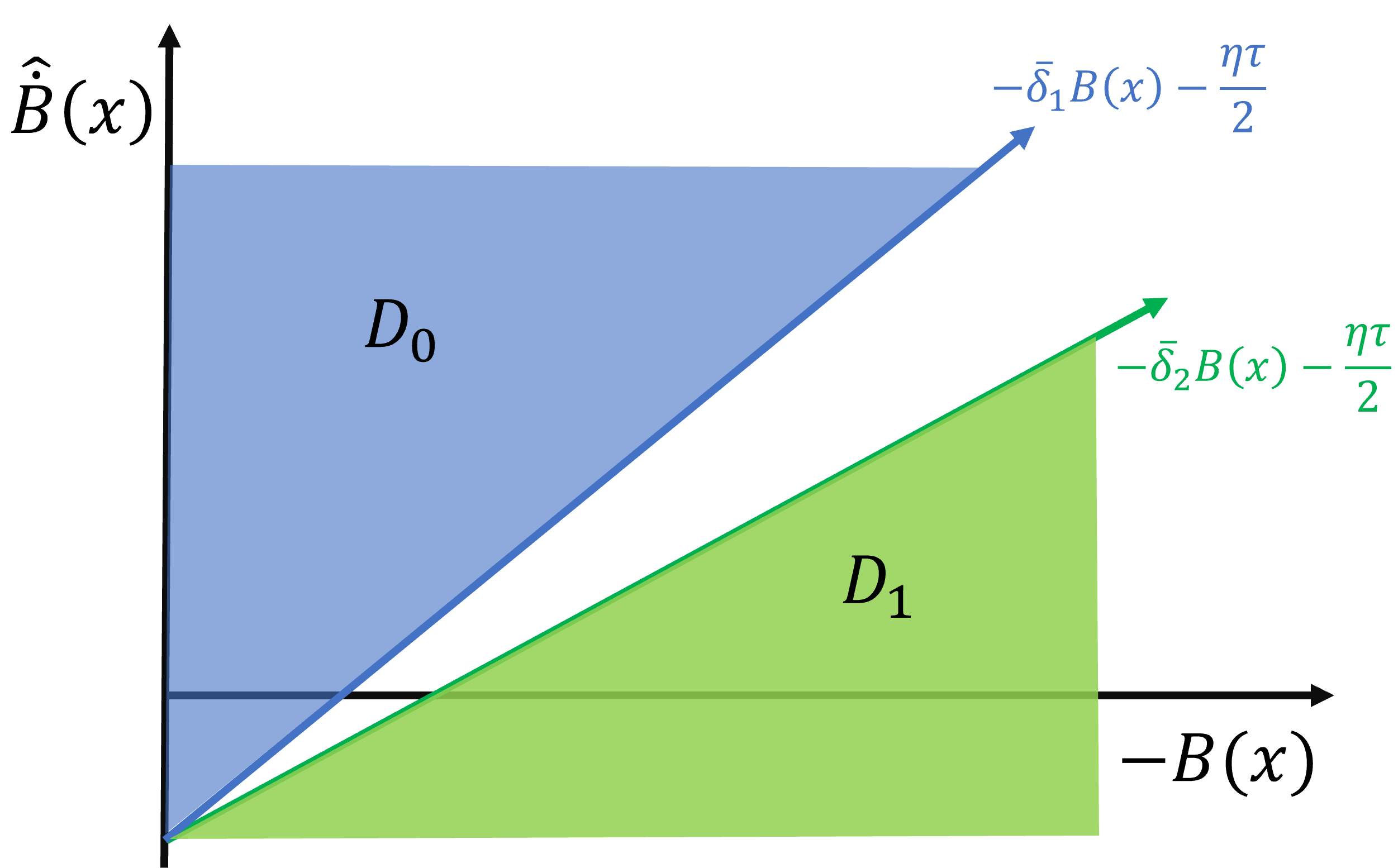}
	\caption{Jump sets $D_0$ and $D_1$ for the hybrid control law.}
	\label{fig:D0 D1 sets}
\end{figure}

\subsection{Hybrid Control Law for Recovery}
In the control assignment given in \eqref{eq: switch input}, we take a somewhat conservative approach and assume that the attack duration is for the maximum possible length $\overline T$. 
This conservatism can be relaxed by using the following control assignment: 
\begin{align}\label{eq: switch input new}
    u(t,x) = \begin{cases}(\lambda_v(x), \lambda_s(x)) & \textrm{if}\quad  t\notin \mathcal T_a \bigcup {\hat {\mathcal T}}_d,\\
    (u_a(t), \lambda_s(x)) & \textrm{if} \quad t\in \mathcal T_a\setminus {\hat {\mathcal T}}_d, \\
    (u_a(t), k_s(x)) & \textrm{if} \quad t\in \mathcal T_a \bigcap {\hat {\mathcal T}}_d, \\
    (\lambda_v(x), k_s(x)) & \textrm{if} \quad t\in {\hat {\mathcal T}}_d \setminus \mathcal T_a, \\
    \end{cases}
\end{align}
where 
\begin{align}\label{eq: T_d estimate set}
    {\hat {\mathcal T}}_d = \left\{t \; \Big |\; \hat{\dot B}(x(t),\tau) > \delta B(x(t), \tau)-\frac{\eta\tau}{2}, x(t)\in K\setminus\textnormal{int}(K_{\bar c})\right\},
\end{align}
is the set of times when the system detects an attack. However, the above switching law can potentially lead to Zeno behavior {due to the control input $u_s$ oscillating between $\lambda_s(x)$ and $\kappa_s(x)$ at the switching surface $D_x = \{x\; |\; \hat{\dot B}((x(t), \tau) = \gamma(t) - \frac{\eta\tau}{2}\}$. To avoid Zeno, inspired by the hybrid control strategy in \cite{wintz2022global}, we define a hybrid control law for the safe input $u_s$ with a \textit{hysteresis}. Consider the following sets:
\begin{align*}
    D_{0,x} & = \Big\{x\; \Big |\; \hat{\dot B}(x)\geq -\bar \delta_1B(x)-\frac{\eta\tau}{2}, x\in K\setminus\textnormal{int}(K_{\bar c})\Big\}, \\
    D_{1,x} & = \Big\{x\; \Big |\; \hat{\dot B}(x)\leq -\bar \delta_2B(x)-\frac{\eta\tau}{2}, x\in K\setminus\textnormal{int}(K_{\bar c})\Big\} 
\end{align*}
with $\bar\delta_2<\bar\delta_1$. Figure \ref{fig:D0 D1 sets} illustrates the sets $D_{0,x}$ and $D_{1,x}$, and the buffer zone between the two sets that would help avoid Zeno behavior. Instead of switching on the set $D_x$, suppose the input $u_s$ switches from $\lambda_s(x)$ to the recovery feedback $k_s(x)$ if the system state $x$ is in the set $D_{0,x}$ and it switches back to the nominal feedback $\lambda_s(x)$ when $x\in D_{1,x}$. Since the sets $D_{0,x}$ and $D_{1,x}$ are closed and disjoint for any $\bar \delta_2<\bar\delta_1$ (see Figure \ref{fig:D0 D1 sets}), the Zeno solutions are not possible. Based on this, define the sets:
\begin{subequations}\label{eq: sets D0 D1}
\begin{align}
    D_0 & = \Big\{z\; \Big |\; \hat{\dot B}(x)\geq -\bar \delta_1B(x)-\frac{\eta\tau}{2}, x\in K\setminus\textnormal{int}(K_{\bar c}), q = 0\Big\}, \\
    D_1 & = \Big\{z\; \Big |\; \hat{\dot B}(x)\leq -\bar \delta_2B(x)-\frac{\eta\tau}{2}, x\in K\setminus\textnormal{int}(K_{\bar c}), q = 1\Big\}, 
\end{align}
\end{subequations}
and 
\begin{align}
    C_0 & = \overline{\Big(\mathbb R^n\times \{0\}\Big)\setminus D_0}, \quad 
    C_1 = \overline{\Big(\mathbb R^n\times \{1\}\Big)\setminus D_1}.
\end{align}
Let $q\in Q\coloneqq \{0, 1\}$ be a logic variable that follows a hybrid dynamics given as
\begin{subequations}\label{eq: logic var dyn}
\begin{align}
    \dot q & = 0 \quad \quad \; \; \hspace{1pt} z\in C, \\
    q^+ & = 1-q \quad  z\in D, 
\end{align}
\end{subequations}
where $z \coloneqq (x,q)\in \mathbb R^{n}\times Q$ is the state of the augmented system. The sets $C$ and $D$ are defined as
\begin{subequations}
\begin{align}
    D & = D_0 \cup D_1,\\ 
    C & = C_0\cup C_1.
\end{align}
\end{subequations}
Given feedback laws $\lambda_s$ and $k_s$, the hybrid control law for the safe input $u_s$ is defined as
\begin{align}\label{eq: hybrid input new}
    u_s(z) & = \begin{cases}\lambda_s(x) & \textnormal{if} \quad (x, q)\in C_0,\\
    k_s(x) & \textnormal{if}\quad (x,q)\in C_1.\end{cases}
\end{align}
Next, we show that Zeno is not possible with the hybrid control law \eqref{eq: hybrid input new}. To this end, let $\{t_j\}_{j = 0}^J$ denote the sequence of jump times with $t_0 = 0$ and $t_{j+1}\geq t_j$, $j\geq 0$. }
\begin{Lemma}
Assume that the functions $\lambda_s, k_s:\mathbb R^n\rightarrow\mathcal U_s$ are continuous. Then, there exists $\zeta>0$ such that $t_{j+1}-t_j\geq \zeta$ for each $j\geq 0$.  
\end{Lemma}
\begin{proof}
The proof is based on showing that $D\cap g(D) = \emptyset$, where $g:\mathbb R^n\times Q\rightarrow\mathbb R^n\times Q$ is the jump dynamics given as
\begin{align}
    z^+ = g(z) \coloneqq \begin{bmatrix}
    g_x(z)\\ g_q(z)
    \end{bmatrix} \coloneqq \begin{bmatrix}
    x\\ 1-q
    \end{bmatrix} \quad z\in D.
\end{align}
For any $z_0 = (x_0, 0)\in D_0$, it holds that $\hat{\dot B}(x_0)\geq -\bar\delta_1 B(x_0)-\frac{\eta\tau}{2}$ and $B(x_0)<0$. Now, consider $z = g(z_0)$. Since $g_x(x) = x$ for each $x\in D$, it holds that $\hat{\dot B}(x)\geq -\bar\delta_1 B(x)-\frac{\eta\tau}{2}$. With $\bar \delta_2<\bar \delta_1$, it holds that $-\bar \delta_2B(x)<-\bar \delta_1 B(x)$. Thus, we have that $\hat{\dot B}(x)\geq -\bar\delta_1 B(x)-\frac{\eta\tau}{2}>-\bar \delta_2 B(x)-\frac{\eta\tau}{2}$, and hence, $g(z_0)\notin D_1$. Conversely, for any $z_1 = (x_1, 1)\in D_1$, it holds that $\hat{\dot B}(x_1) \leq -\bar\delta_2 B(x_1)-\frac{\eta\tau}{2}$. For any $x = g(x_1)$, it holds that $\hat{\dot B}(x) \leq -\bar\delta_2 B(x)-\frac{\eta\tau}{2}< -\bar\delta_1B(x)-\frac{\eta\tau}{2}$ and thus, $g(z_1)\notin D_0$. Hence, $D\cap g(D) = \emptyset$. 

Furthermore, note that the set $K\setminus\textnormal{int}(K_{\bar c})$ is closed, and the functions $B$ and $\hat{\dot B}$ are continuous. Thus, the sets $D_0$ and $D_1$ are closed, and consequently, the set $D$ is also closed. The sets $C_1$ and $C_2$ are closed by definition, and hence, the set $C$ is also closed. The function $F$ is continuous under the conditions of the lemma and the function $g\in \mathcal C^0$, and so, it satisfies the \textit{hybrid basic conditions} (i.e., conditions (A0)-(A3) in \cite{sanfelice2007invariance}). 
 
It remains to be shown that the system trajectories remain bounded. Consider $z\in C_0$. In this case, by definition, $x\in K$ and $\hat {\dot B}\leq -\bar \delta_1B(x)-\frac{\eta\tau}{2}$, which implies that \eqref{eq: safety cond} holds, and hence, the system trajectories do not leave the set $K$. Next, for $z\in C_1$, per \eqref{eq: hybrid input new}, the control input is defined as $u_s = k_s(x)$. Under Assumption \ref{assum: Sc set}, the feedback $k_s$ can render any sublevel set of $B$ in $K\setminus\textnormal{int}(K_{\bar c})$ forward invariant, and thus, the system trajectories do not leave the set $K$. The system state $x$ is continuous on $D$, and hence, we have that $x(t)\in K$ for all times, and with $K$ being compact, the system trajectories are bounded. Hence, from \cite[Lemma 2.7]{sanfelice2007invariance}, it holds that there exists $\zeta>0$ such that $t_{j+1}-t_j\geq \zeta$ for each $j\geq 0$. 
\end{proof}

Thus, there is a non-zero dwell time $\zeta$ between jump times that rules out any Zeno behavior. The closed-loop dynamics under the hybrid control law \eqref{eq: hybrid input new} and attack model \eqref{eq: attack input} is given as\footnote{We omit the argument $(t,j)$ from the functions $z$ and $x$ for the sake of brevity.}
\begin{align}\label{eq: hyb closed loop}
    \mathcal H: \begin{cases}\dot z = \begin{bmatrix}F(x,(u_a(t), u_s(z))\\ 0
    \end{bmatrix} &  t\in \mathcal T_a, z\in C\vspace{3pt}\\ 
    \dot z = \begin{bmatrix}F(x,(\lambda_v(x), u_s(z))\\ 0
    \end{bmatrix} & t\notin \mathcal T_a, z\in C\\
    z^+ = \begin{bmatrix}
    x \\ 1-q
    \end{bmatrix} & z\in D.
    \end{cases}
\end{align}
The following corollary to Theorem \ref{thm: suff safety under attack} holds for the hybrid closed-loop system \eqref{eq: hyb closed loop}. 

\begin{Corollary}\label{cor: suff safety under attack hybrid}
Given system $\mathcal S$ with $F\in \mathcal C^1$, $B\in \mathcal C^2$ and the attack model \eqref{eq: attack model}, suppose that Assumption \ref{assum: d bound} holds, and Assumptions \ref{assum: dot B bar c}-\ref{assum: Sc set} hold for some $\bar c\in (0, c_M)$. 
Then, there exist feedback laws $\lambda:\mathbb R^n\rightarrow\mathcal U$ and $k_s:\mathbb R^n\rightarrow\mathcal U_s$ such that under the effect of the input $u$ in \eqref{eq: hybrid input new}, 
the system trajectories of \eqref{eq: hyb closed loop} satisfy $z(t, j)\in K\times \{0, 1\}$ for all $(t, j)\in \textnormal{dom}~z$ and for all $z(0, 0)\in X_0\times \{0, 1\} = \textnormal{int}(K)\times \{0, 1\}$. 
\end{Corollary}
\begin{proof}
{The proof follows from the similar arguments used in the proof of Theorem \ref{thm: suff safety under attack}. For any $j\in \mathbb N$ such that $(t, j) \in \textnormal{dom}~z$, consider the cases: Case 1: $t\in \mathcal T_a$ and $z(t, j)\in C_0$, Case 2: $t\in \mathcal T_a$ and $z(t, j)\in C_1$, Case 3: $t\notin \mathcal T_a$ and $z(t, j)\in C_1$ and Case 4: $t\notin \mathcal T_a$ and $z(t, j)\in C_0$ (similar to Case 1, Case 2, Case 3 and Case 4 in Theorem \ref{thm: suff safety under attack}, respectively).}
\end{proof}


We note that the main challenge with the proposed method for synthesizing the hybrid control law is finding parameter $\bar c$ for the satisfaction of Assumptions \ref{assum: dot B bar c} or \ref{assum: Sc set}. While Assumptions \ref{assum: dot B bar c} and \ref{assum: Sc set} serve different purposes (as illustrated in the proof of Theorem \ref{thm: suff safety under attack}), it is easy to see that satisfaction Assumption \ref{assum: Sc set} for some $\bar c\in (0,c_M)$ implies Assumption \ref{assum: dot B bar c} holds for the same $\bar c$. Thus, it is sufficient to verify that Assumption \ref{assum: Sc set} holds. One practical method of finding a subset of the safe set $K$, where Assumption \ref{assum: Sc set} holds, is the computationally efficient sampling-based method proposed in \cite{garg2021sampling}. Note that the sampling-based method in \cite{garg2021sampling} relies on a specific sampling method, known as \textit{triangulation} of a sphere. In the next section, we propose a new sampling-based method for computing a subset $K_c\subset K$ such that \eqref{eq: Assum Sc attack} holds for each $x\in K_c\setminus K_{\bar c}$ for some $\bar c\in (c, c_M)$. In the simulation results, we illustrate that the proposed sampling method is faster than the triangulation method used in \cite{garg2021sampling}. 

\section{Sampling Method for Safety}\label{sec: sampling}
Let us consider the forward invariance of the set $K$. Verifying that a set $K\subset \mathbb R^n$ is forward invariant involves checking the inequality \eqref{eq: safety cond} for each $x\in \partial K$, which is a $(n-1)-$ dimensional manifold and also has infinitely many points. Next, consider the inequality \eqref{eq: Assum Sc attack} that needs to be checked for each $K\setminus K_{\bar c}$ for forward invariance of the set $K$ under attacks. 
In this section, we present a sampling-based method to verify such inequalities in a computationally efficient method. In particular, we consider the inequality
\begin{align}
    H(x)\leq 0 \quad \forall x\in \mathcal A_H,
\end{align}
for an appropriate function $H:\mathbb R^n\rightarrow\mathbb R$ and an appropriate set $\mathcal A_H\subset \mathbb R^n$. First, we start with a sampling-based method of verifying an inequality on the boundary of a compact set $K\subset \mathbb R^n$ given as $K = \{x\;|\; B(x)\leq 0\}$ for a sufficiently smooth function $B$. 

\subsection{Sampling-based Method for Forward Invariance without Attacks}\label{sec: FI no attack}

In this section, we present a sampling-based method to compute a set that is forward invariant for \eqref{eq: actual system} when there is no attack. In this case, the function $H$ is defined as
\begin{align}\label{eq: H safety}
    H(x) = \inf_{u\in \mathcal U}\sup_{\zeta\in F(x, u)}L_\zeta B(x,u) +l_B\delta, 
\end{align}
where $l_B$ is the Lipschitz constant of the function $B$, and $\delta$ is as defined in Assumption \ref{assum: d bound}. we propose a sampling-based method of checking a modification of \eqref{eq: H safety} on a finite set of points on $\partial K$ so that conditions of Lemma \ref{lemma: nec suff safety} are satisfied.


We start by making the following assumption on the regularity of the function $H$ defined in \eqref{eq: H safety}.
\begin{Assumption}\label{assum: LC F del B}
The function $H$
is Lipschitz continuous on $K$ with constant $l_H>0$. 
\end{Assumption}

We make the following assumption on the sampling points $\{x_i\}_{\mathcal I}$.  
\begin{Assumption}\label{assum: max dist tring nD}
Given $c\in [0,  c_M)$, the sampling points $\{x_i\}_{\mathcal I}$ and $d_a\in \begin{bmatrix}0,  d_{M,n}\end{bmatrix}$, for each $x\in \partial K_c$, there exists $y\in \{x_i\}_{\mathcal I}$ such that
\begin{align}\label{eq: xj xk dist cond nd}
d_{K_c}(x,y)\leq \frac{d_a}{2},
\end{align}
where $d_{K_c}(x,y)$ denotes the shortest arc-length between the points $x,y\in \partial K_c$.
\end{Assumption}

\noindent Now, we show that 
if the following holds
\begin{align}\label{eq: sup H cond 3D}
    H(x_i)\leq-l_H\frac{d_a}{2} \quad \forall i\in \mathcal I,
\end{align}
where $l_H$ is as defined in Assumption \ref{assum: LC F del B},
then, \eqref{eq: H safety} holds on the boundary $\partial K_c$. 

\begin{Theorem}\label{thm: sufficient sup H cond}
Suppose that the function $H$ defined in \eqref{eq: H safety} satisfies Assumption \ref{assum: LC F del B}. Given $c\in [0, c_M)$, $d_a\in \begin{bmatrix}0,  d_{M,n}\end{bmatrix}$, and the sampling points $\{x_i\}_{\mathcal I}\subset \partial K_c$,
if Assumption \ref{assum: max dist tring nD} and \eqref{eq: sup H cond 3D} hold, then, \eqref{eq: safety cond} holds.
\end{Theorem}
\begin{proof}

Using Assumption \ref{assum: LC F del B} and \eqref{eq: sup H cond 3D}, it holds that
\begin{align*}
   H(\bar x) \leq & H(x) + l_H|\bar x-x|\\
    \leq & -l_H\frac{d_a}{2} + l_H|\bar x-x|,
\end{align*}
for all $x, \bar x\in \partial K_c$. Under Assumption \ref{assum: max dist tring nD}, for every $\bar x\in \partial K_c$, there exists $y\in \{x_i\}_{\mathcal I}$ such that $d_{K_c}(\bar x, y)\leq \frac{d_a}{2}$. Thus, substituting $x = y$ in the above inequality,
we obtain that $H(\bar x)\leq -l_H\frac{d_a}{2} + l_H \frac{d_a}{2} = 0$ for all $\bar x\in \partial K_c$, which completes the proof.
\end{proof}

Thus, the inequality \eqref{eq: sup H cond 3D} can be checked at finitely many points to verify the inequality \eqref{eq: safety cond} for forward invariance of the set $K_c$.

\subsection{Sampling-based Method for Forward Invariance under Attacks}\label{sec: FI attack}

In this section, we propose a method of verifying the inequality \eqref{eq: Assum Sc attack} using a sampling-based method. In particular, we discuss how to modify the inequality \eqref{eq: Assum Sc attack} to facilitate the sampling-based method. For a given $c\in [0, c_M)$ and $\bar c\in (c, c_M)$, define
\begin{align}\label{eq: d c x c}
    (x_c, y_c) & = \arg\max_{x\in K_c}\min_{y\in K_{\bar c}}d_{S}(x,y), \\
    d_c & = d_S(x_c, y_c),
\end{align}
so that the maximum arc-length distance between the sets $K_c$ and $K_{\bar c}$ is $d_c$ at points $x_c\in K_c$ and $y_c\in K_{\bar c}$. Next, consider $\hat c\in (c, \bar c)$ and a set of sampling points $\{x_i\}_{\mathcal I}$ from the set $\{x\; |\; B(x)\leq -\hat c\}$ satisfying Assumption \ref{assum: max dist tring nD} (with $c = \hat c$ in Assumption \ref{assum: max dist tring nD}) for a given $d_a>0$. To this end, define function $H$ as
\begin{align}\label{eq: H safety attack}
    H(x) = \sup\limits_{u_v\in \mathcal U_v}\inf\limits_{u_s\in \mathcal U_s} \sup_{\zeta\in F(x, (u_v,u_s))}L_\zeta B(x, (u_v, u_s))+l_B\delta.
\end{align}
Using this, we show that if the following holds
\begin{align}\label{eq: sup H cond 3D attack}
    H(x_i)\leq-l_H\left(\frac{d_a}{2}+d_c\right) \quad \forall i\in \mathcal I,
\end{align}
where $l_H$ is as defined in Assumption \ref{assum: LC F del B},
then, \eqref{eq: Assum Sc attack} holds on $K_c\setminus \textnormal{int}(K_{\bar c})$. 
Similar to Theorem \ref{thm: sufficient sup H cond}, we can state the following result. 
\begin{Theorem}\label{thm: sufficient sup H cond attack}
Suppose that the function $H$ defined in \eqref{eq: H safety attack} satisfies Assumption \ref{assum: LC F del B}. Given $c\in [0, c_M)$, $\bar c\in (c, c_M)$ $d_a\in \begin{bmatrix}0,  d_{M,n}\end{bmatrix}$, and the sampling points $\{x_i\}_{\mathcal I}\subset \partial K_{\hat c}$ for some $\hat c\in (c,\bar c)$,
if \eqref{eq: sup H cond 3D attack} holds and Assumption \ref{assum: max dist tring nD} holds with $c = \hat c$, then \eqref{eq: Assum Sc attack} holds for each $x\in K_c\setminus \textnormal{int}(K_{\bar c})$.
\end{Theorem}

\begin{proof}
Using the Lipschitz continuity of the function $H$ with constant $l_H>0$, we have that
\begin{align}\label{eq: proof lH c bar c}
    H(x)\leq & H(y) + l_H|x-y|.
\end{align}
It holds that for each $i\in \mathcal I$, $H(x_i) \leq -l_H \left( \frac{d_a}{2}+d_c\right)$. Consider any point $x\in K_c\setminus\textnormal{int}(K_{\bar c})$. Per \eqref{eq: d c x c}, it holds that there exists $\hat x\in S_{\hat c}$ such that $d_S(x, \hat x)\leq d_c$. Furthermore, per Assumption \ref{assum: max dist tring nD}, it holds that for $\hat x\in S_{\hat c}$, there exists $x_i\in \{x_i\}_{\mathcal I}$ such that $d_S(\hat x, x_i)\leq \frac{d_a}{2}$. Thus, we obtain that for each $x\in K_c\setminus \textnormal{int}(K_{\bar c})$, there exists $y\in \{x_i\}_{\mathcal I}$ such that $d_S(x, y)\leq \frac{d_a}{2}+d_c$. Using this in \eqref{eq: proof lH c bar c} along with the fact that $y\in \{x_i\}_{\mathcal I}$, we obtain that
\begin{align*}
    H(x)& \leq  -l_H \left( \frac{d_a}{2}+d_c\right)
+ l_H|x-y| \\
 &\leq  -l_H \left( \frac{d_a}{2}+d_c\right)
+ l_Hd_S(x,y)\leq 0,
\end{align*}
for each $x\in K_c\setminus\textnormal{int}(K_{\bar c})$, which completes the proof. 
\end{proof}

Thus, the inequality \eqref{eq: Assum Sc attack} on a set $K_c\setminus\textnormal{int}(K_{\bar c})$ can be verified by checking a modified inequality on a finitely many sampling points. Note that the results in the previous section are a special case of the results in this section with $c = \bar c$ (resulting in $d_c = 0$). First, we demonstrated how to verify an inequality on the boundary of a set. Then, we demonstrated how to verify an inequality on a \textit{buffer} zone on the boundary of a set. 



\subsection{Sampling of Higher-dimensional Sets}
In this section, we propose a method of sampling the boundary of a general set $K\in \mathbb R^n$ such that Assumption \ref{assum: max dist tring nD} holds. To this end, we use the sampling points from a lower-dimensional set to a higher-dimensional set. Thus, starting from the 2-sphere, we obtain sampling points for the 3-sphere, using which, we obtain sampling points for the 4-sphere, and we repeat the process till we reach $(n-1)$-sphere. Then, we discuss how to sample the boundary of the set $K$ using the samples on $(n-1)$-sphere. 

To illustrate the idea, we show how to obtain sampling points for the boundary of a 2-sphere using the sampling points from a 1-sphere, i.e., a circle. Note that the boundary of a 2-sphere $\partial S_2 = \{(x_1,x_2,x_3)\; |\; x_1^2+x_2^2+x_3^2 = 1\}$ can be parameterized as
\begin{align}
    \partial S_2 = \Big\{(x_1,x_2,x_3)\; \Big|\; & x_1 = \cos(\phi)\sin(\theta), x_2 = \sin(\phi)\sin(\theta), \nonumber \\
    & x_3 = \cos(\theta), \phi\in [0, 2\pi], \theta\in [0, pi]\Big\}.
\end{align}
Observe that for a fixed $\phi\in [0, 2\pi]$, the resulting set is 1-sphere. For a given $d_1>0$, assume that $\{(x_1, x_2)_i\}_{\mathcal I_1}$ is the set of sampling points on the boundary of 1-sphere ($\partial S_1 = \{(x_1, x_2)\; |\; x_1^2+x_2^2 = 1\}$ such that for each $x\in \partial S_1$, there exists $\bar x\in \{(x_1, x_2)_i\}_{\mathcal I_1}$ such that $d_{S_1}(x, \bar x)\leq \frac{d_1}{2}$. In other words, for each $x\in \{(x_1, x_2)_i\}_{\mathcal I_1}$, there exists $y\in \{(x_1, x_2)_i\}_{\mathcal I_1}$ such that $d_{S_1}(x,y)\leq d_1$. For a fixed $\phi$, the following assignment 
\begin{align*}
    \bar x_1 = x_1, \quad  \bar x_2 = x_2\cos(\phi), \quad \bar x_3 = x_2\sin(\phi),
\end{align*}
defines a set of sampling points on the boundary of $K_2$, confined to $K_2\cap S_1$. Let us sample the parameter $\phi$ in $[0, 2\pi]$ with step $d_{1,1}$, i.e., $\phi(i+1)-\phi(i) = d_a$, with $\phi(1) = 0$ and $\phi_N = 2\pi$, where $N = \left\lceil\frac{2\pi}{d_a}\right\rceil$ is the number of sampling points of $\{\phi_i\}$. The sampling points for $\partial S_2$ can be defined as
\begin{align}\label{eq: s1 s2 sampling}
    \{x_i\}_{\mathcal I_3} = \Big\{(\bar x_1, \bar x_2, \bar x_3)\; \Big | \;& \bar x_1 = x_1, \bar x_2 = x_2\cos(\phi), \nonumber\\
    & \bar x_3 = x_2\sin(\phi), \phi\in \{\phi_i\} \Big\}.
\end{align}

\begin{figure}[t]
	\centering
	\includegraphics[width=1\columnwidth,clip]{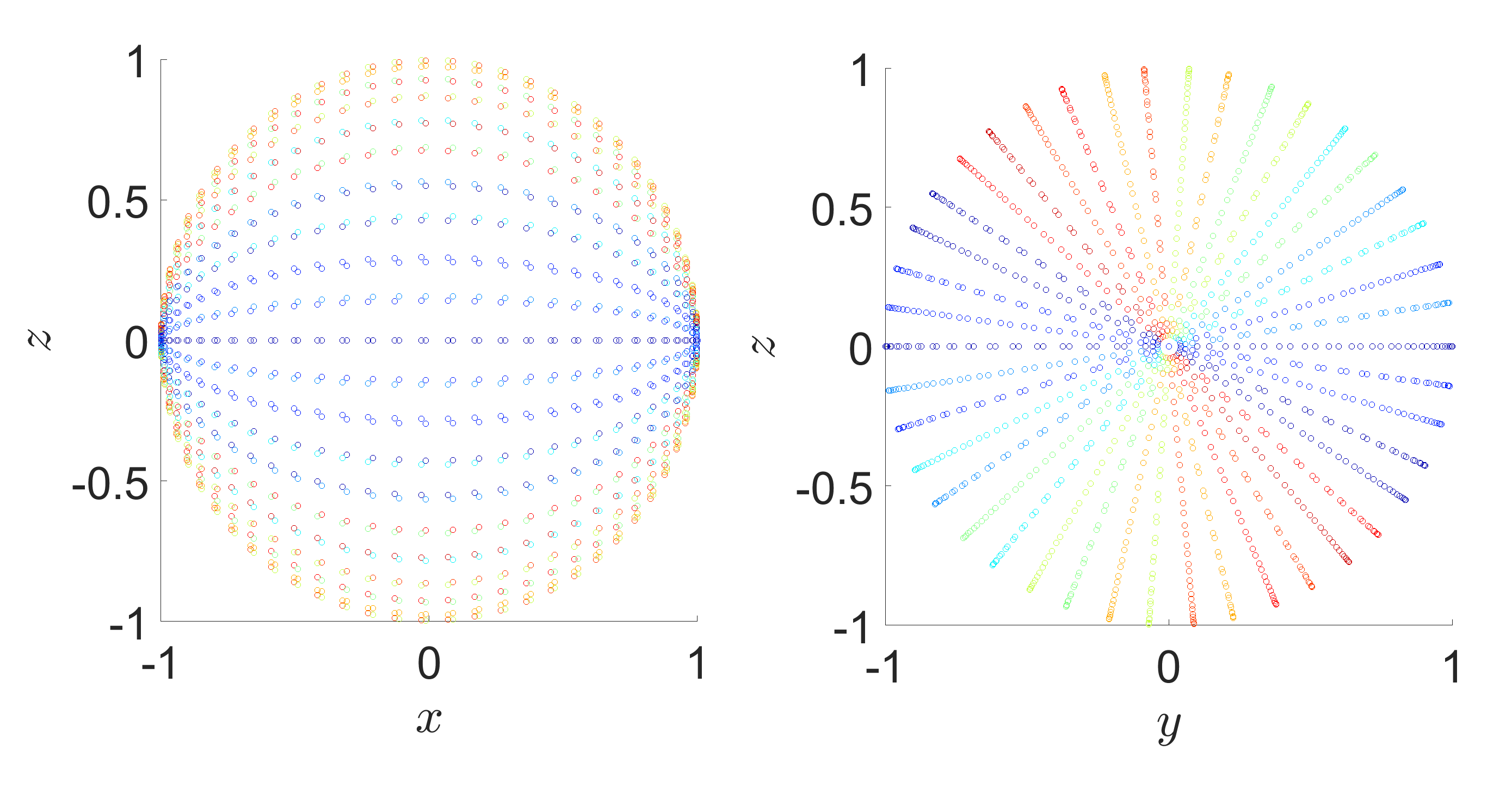}
	\caption{Sampling of 2-sphere using the sampling points of 1-sphere.} 
	\label{fig: 2d to 3d sampling}
\end{figure}

\begin{Lemma}\label{lemma: S1 to S2 samp}
For each $x\in \partial S_2$, there exists $\bar x\in \{x_i\}_{\mathcal I_3}$, where $\{x_i\}_{\mathcal I_3}$ is defined in \eqref{eq: s1 s2 sampling}, such that $d_{S_2}(x,\bar x)\leq \frac{d_1}{2}+\frac{d_{1,1}}{2}$. 
\end{Lemma}
\begin{proof}
For a given $\phi\in \{\phi_i\}$, define $\partial S_1(\phi) = \{(\bar x_1, \bar x_2, \bar x_3)\; |\; \bar x_1 = x_1,\bar x_2 = x_2\cos(\phi),\bar x_3 = x_2\sin(\phi)\}$. Consider the two cases: $x\in \bigcup\limits_{i = 1}^N\partial S_1(\phi_i)$ and $x\notin \bigcup\limits_{i = 1}^N\partial S_1(\phi_i)$. In the first case, under the assumption on the sampling points $\{(x_1, x_2)\}_{\mathcal I_1}$, it holds that there exists $\bar x\in \{x_i\}_{\mathcal I_2}$ such that $d_{S_1}(x, \bar x)\leq \frac{d_1}{2}$. Note that $d_{S_1}(x, \bar x) = d_{S_2}(x, \bar x)$, and thus, it holds that for each $x\in \bigcup\limits_{i = 1}^N\partial S_1(\phi_i)$, there exists $\bar x\in \{x_i\}_{\mathcal I_2}$ such that $d_{S_2}(x,\bar x)\leq \frac{d_1}{2}\leq \frac{d_1}{2}+\frac{d_{1,1}}{2}$. 

In the second case, there exists $i\in \{1, 2, \dots, N\}$ such that $x$ lies in the set $D_i = \Big\{(\bar x_1, \bar x_2, \bar x_3)\; \Big | \; \bar x_1 = x_1, \bar x_2 = x_2\cos(\phi), \bar x_3 = x_2\sin(\phi), \phi\in [\phi_i,\phi_{i+1}] \Big\}$, i.e., the ``spherical strip" on $\partial S_2$ bounded by $\partial S_1(\phi_i)$ and $\partial S_1(\phi_{i+1})$. Let $z^* = \argsup\limits_{z\in \partial S_1(\phi_i)\cup\partial S_1(\phi_{i+1})}d_{S_2}(x,z)$. Since $\phi_{i+1}-\phi_i = d_{1,1}$, it holds that $d_{S_2}(x,z^*) = \frac{d_{1,1}}{2}$. Without loss of generality, assume that $z^*\in \partial S_1(\phi_i)$. Per assumption on the sampling points on $\partial S_1$, it holds that there exists $\bar x\in \partial S_1(\phi_i)$ such that $d_(S_1)(z^*, \bar x)\leq \frac{d_1}{2}$, Combining this with $d_s(x,z^*)$, we obtain that $d_{S_2}(x,\bar x)\leq \frac{d_1}{2}+\frac{d_{1,1}}{2}$. 
\end{proof}

Thus, we illustrate that it is possible to sample a sphere in a higher dimension ($\mathbb R^3$) using the sampling points of a sphere in a lower dimension ($\mathbb R^2$). While the proposed method of obtaining the sampling points on a higher dimension is not an optimal one in terms of the number of sampling points, it is efficient in terms of the computational time required to obtain the sampling points. 
As a consequence, for $x\in \{x_i\}_{\mathcal I_2}$, there exists $y\in \{x_i\}_{\mathcal I_2}$ such that $d_{S_2}(x,y)\leq d_1+d_{1,1}$. Let us define \textit{max-min} inter-sampling distance on $(n-1)-$sphere $K_{n-1}\subset \mathbb R^n$ as
\begin{align}
    d_n = \max_{x\in \partial S_{n-1}}\min_{y(x)\in \partial S_{n-1},y(x)\neq x}d_{S_{n-1}}(x,y(x)),
\end{align}
so that for each $x\in \partial S_{n-1}$, there exists $y\in \partial S_{n-1}$ such that $d_{S_{n-1}}(x,y)\leq d_n$. Note that from Lemma \ref{lemma: S1 to S2 samp}, we obtain that $d_2\leq d_1+d_{1,1}$, where $d_{1,1}$ depends on the sampling of the parameter $\phi$. Now, note that the parameterization of $(n-1)-$sphere is given as
\begin{align*}
    x_1 & = \cos(\phi_1), \\
    x_2 & = \sin(\phi_1)\cos(\phi_2), \\
    x_3 & = \sin(\phi_1)\sin(\phi_2)\cos(\phi_3), \\
    \vdots\\
    x_n & = \prod\limits_{i=1}^n\sin(\phi_i).
\end{align*}
Note also that for a fixed value of $\phi_n$, the resulting manifold is $(n-2)-$sphere. Thus, we can use the same construction to obtain sampling points on $K_{n-1}$ from a set of sampling points on $K_{n-2}$. The relation between the max-min inter-sampling distance is given by $d_n \leq d_{n-1}+d_{n-1,n-1}$, where $d_{n-1,n-1}$ is the sampling step for the parameter $\phi_n$. Using this observation, the following relation can be established
\begin{align*}
    d_n & \leq d_{n-1}+d_{n-1,n-1}\\
    & \leq d_{n-2}+d_{n-2,n-2}
+d_{n-1,n-1},\\
&\leq d_1+\sum_{i=1}^{n-1}d_{i,i}
\end{align*}
where $d_{i-1,i-1}$ is step-size used for sampling the parameter $\phi_{i}$ for obtaining sampling points on $(i)-$sphere from $(i-1)-$sphere, $i\geq 2$. Thus, for a required max-min inter-sampling point distance $d_a>0$ such that $d_n\leq d_a$, the parameters $d_1$ and $d_{i,i}$ for $i\in \{1, 2, \dots, n-1\}$ can be chosen so that $d_1+\sum_{i=1}^{n-1}d_{i,i}\leq d_a$. 

Thus, we proposed a method of obtaining sampling points on a higher-dimensional sphere using the sampling points from a unit sphere in $\mathbb R^2$. 

\subsection{Iterative Algorithm}

Note that there are two parameters that can facilitate satisfaction of \eqref{eq: sup H cond 3D} in the following manner:
\begin{itemize}
    \item Parameter $c$: larger value of $c$ results in smaller values of $d_M$, thus, reducing the right-hand side of \eqref{eq: sup H cond 3D}, and making it easier to satisfy it; and
    \item Number of sampling points $N_p$: larger $N_p$ results in smaller value of $d_{M,n}$.
\end{itemize}

Based on these observations, an iterative algorithm can be formulated to check whether there exists a feasible $c$ and $\bar c$, such that \eqref{eq: sup H cond 3D} holds. 

We formulate our algorithm with the following steps:
\begin{itemize}
    \item[1)] For a given value of $0\leq c\leq c_M$, ${\mathcal U}_v$ and number of sampling points $N_p$, sample $\{x_i\}_{\mathcal I}$ from the set $\partial K_c$ and check if \eqref{eq: sup H cond 3D} holds for all the sampling points;
    \item[2)] Increase $N_p$ and repeat steps 1)-3) until \eqref{eq: sup H cond 3D} holds or the maximum value ($N_{max}$) of $N_p$ is reached.
\end{itemize}
Using these steps, we propose Algorithm \ref{algo: iter algo} which returns a feasible $c\in (0, c_M)$ and $\bar c\in (c, c_M)$ such that \eqref{eq: Assum Sc attack} holds for $x\in K_c\setminus\textnormal{int}(K_{\bar c})$. In other words, this algorithm can compute the set of initial conditions $K_c$, and the set of \textit{tolerable} attacked inputs via ${\mathcal U}_v$ such that the system can satisfy the safety property under attacks. The order in which the parameters $c, {\mathcal U}_v$, and $N_p$ are tuned can be changed, which can potentially change the output of the algorithm. 

\begin{Remark}
The computational complexity of Algorithm \ref{algo: iter algo} is only a function of the number of sampling points $N_p$ (which, in principle, is a user-defined parameter) and \textbf{is independent of the non-linearity of the function $F$ or function $B$}. Note that the minimum number of samples required to generate a simplex on an $(n-1)-$sphere in $\mathbb R^n$ is $(n+1)$, and hence, the initial sampling number $N_{c0}$ in Algorithm \ref{algo: iter algo} \textbf{is linear in the dimension $n$}.\footnote{Note that for the re-sampling step, the initial set of samples $\{x_i\}_{\mathcal I}$ can be used. In particular, for every face $\mathcal X_j$ consisting of points $\{x_j\}$, a new sampling point can be defined as $\bar x_j = x_o + \frac{\tilde x_j-x_o}{|\tilde x_j-x_o|}$ where $\tilde x_j = \frac{1}{n}\sum x_{j_i}$. Since the number of faces in a simplex is linear in $n$, increasing the sampling number has linear computational complexity in $n$.} Thus, unlike reachability based tools in \cite{choi2021robust} where the computational complexity grows exponentially with the system dimension $n$, or SOS based tools \cite{wang2018permissive} that are only applicable to a specific class of systems with linear or polynomial dynamics, Algorithm \ref{algo: iter algo} can be used for general nonlinear system with high dimension.  
\end{Remark}

\begin{Remark}\label{remark: sample general S}
Note that if the set $K$ is convex and $B$ is continuously differentiable, it is diffeomorphic to an $(n-1)-$unit sphere.
Furthermore, when $K$ (equivalently, set $K_c$ for any $c\in (0, c_M)$) is diffeomorphic to an $(n-1)-$unit sphere under a known map $\phi: K \rightarrow \mathcal S_1$, where $\mathcal S_1\subset \mathbb R^n$ is an $(n-1)-$unit sphere, the sampling points on the boundary of the set $K_c$ can be obtained as follows:
\begin{itemize}
    \item[1)] For a given $d_a\in [0, d_{M,n}]$ for sampling on $K_c$, define the corresponding parameter $\bar d_a$ for sampling on $\mathcal S_1$ as 
    \begin{align}
        \hspace{-20pt}\bar d_a \coloneqq \inf_{x, y\in \mathcal S_1}\{d_{\mathcal S_1}(x,y) \; |\; d_{K_c}(\phi^{-1}(x), \phi^{-1}(y)) \geq d_a\}
    \end{align}
    \item[2)] Obtain sampling points $\{\bar x_i\}_{\mathcal I}$ on $\mathcal S_1$ using $\bar d_a$;
    \item[3)] Define sampling points $\{x_i\}_{\mathcal I}$ on $K_c$ as $x_i \coloneqq \phi^{-1}(\bar x_i)$. 
\end{itemize}
\end{Remark}



\begin{algorithm}[t]\label{algo: iter algo}
\SetAlgoLined
\KwData{$F,B,\mathcal U_v, \mathcal U_s, d_a, \epsilon, \varepsilon_1, \varepsilon_2, \delta, \delta_M, N_{max}, N_c, N_{c0}, \gamma_M$}
\textbf{Initialize: }$c = 0, N_p = N_{c0}, \bar \delta = 0$\;
 \While{$c<c_M$}{
         $\bar c = 0 ;$
         \\ \While{$\bar c<c$}{\While{$N_p< N_{max}$}
{
     Sample $\{x_i\}_{\mathcal I}$ from $\{B(x) = -\bar c\}$\;
     \While{$\bar \delta<\delta_M$}{
      \lIf{$\mathcal I_{\bar c} \neq \emptyset$\label{line: h s U}}{
      
      $\bar \delta = \bar \delta +\delta$
      }
      }{
      $N_p = N_p+N_c$; \\
      $\bar \delta = 0$;
      }
      }
      $\bar c= \bar c+\varepsilon_1$;\\
      $N_p = N_0$;
     }
     $c = c+\varepsilon_2$;
     }
 \textbf{Return: }$\bar \delta, c, \bar c$\;
 \caption{Iterative method for computing $c$}
\end{algorithm}
  
Thus, the output of Algorithm \ref{algo: iter algo} returns a $(\bar \delta, c, \bar c)$ such that \eqref{eq: Assum Sc attack} holds for $x\in K_c\setminus\textnormal{int}(K_{\bar c})$.


\section{QP-based Recovery Control Synthesis}\label{sec: QP control}


Next, we present a control syntheses method to design both the nominal feedback $\lambda$ and the safe recovery feedback-law $k_s$ for \eqref{eq: switch input}. 
In order to use a tractable optimization problem for control synthesis, we assume that the system \eqref{eq: actual system} is control affine and is of the form
\begin{align}\label{eq: cont affine}
    \dot x = f(x) + g(x)u + d(t,x),
\end{align}
where $f:\mathbb R^n\rightarrow\mathbb R^n$ and $g:\mathbb R^n\rightarrow\mathbb R^{n\times m}$ are continuous functions. Assume that the input constraint set $\mathcal U$ is given as $\mathcal U = \{u\; |\; Au\leq b\}$.

First, we present a quadratic program (QP) formulation to synthesize the nominal feedback law $\lambda$. Consider the following QP for each $x\in K$:
\small{
\begin{subequations}\label{QP gen}
\begin{align}
\hspace{10pt}\min_{(v, \eta)} \quad \frac{1}{2}|v|^2 + & \frac{1}{2}\eta^2\\
    \textrm{s.t.} \;\quad \; \; Av  \leq &  \; b, \label{C1 cont const}\\
    L_fB(x) + L_{g}B(x)v \leq & -\eta B(x)-l_B\delta,\label{C3 safe const}
\end{align}
\end{subequations}}\normalsize
where $q>0$ is a constant, $l_B$ is the Lipschitz constants of the function $B$. Next, we use a similar QP to compute the safe feedback-law $k_s$. To this end, let $g = [g_s\; g_v]$ with $g_s:\mathbb R^n\rightarrow\mathbb R^{n\times m_s}$, $g_v:\mathbb R^n\rightarrow\mathbb R^{n\times m_v}$ and assume that the input constraint set for $u_s$ is given as $\mathcal U_s = \{u_s\; |\; A_su_s\leq b_s\}$. Now, consider the following QP for each $x\in K\setminus \textnormal{int}(K_{\bar c})$:
\small{
\begin{subequations}\label{QP safe k}
\begin{align}
\hspace{10pt}\min_{(v_s, \zeta)} \quad \frac{1}{2}|v_s|^2 + & \frac{1}{2}\zeta^2\\
    \textrm{s.t.} \;\quad \; \; A_sv_s  \leq &  \; b_s, \label{C1 cont const k}\\
    L_{f}B(x) + L_{g_s}B(x)v_s \leq & -\zeta B(x)- l_B\delta\nonumber\\
    & -\sup_{u_v\in \mathcal U_v}L_{g_v}B(x)u_v,\label{C3 safe const k}
\end{align}
\end{subequations}}\normalsize

Let the solution of the QP \eqref{QP gen} be denoted as $(v^*, \eta^*)$ and that of \eqref{QP safe k} as $(v_s^*, \zeta^*)$. In order to guarantee continuity of these solutions with respect to $x$, we need to impose the strict complementary slackness condition (see \cite{garg2019prescribedTAC}). In brief, if the $i-$the constraint of \eqref{QP gen} (or \eqref{QP safe k}), with $i \in \{1, 2\}$, is written as $G_i(x,z)\leq 0$, and the corresponding Lagrange multiplier is $\lambda_i\in \mathbb R_+$, then strict complementary slackness requires that $\lambda_i^*G(x,z^*)<0$, where $z^*, \lambda_i^*$ denote the optimal solution and the corresponding optimal Lagrange multiplier, respectively. We are now ready to state the following result.  

\begin{Theorem}
Given the functions $F, d, B$ and the attack model \eqref{eq: attack model}, suppose Assumptions \ref{assum: d bound}-\ref{assum: Sc set} hold with $\bar \delta>0$ and $\bar c\in (0, c_M)$.
Assume that the strict complementary slackness holds for the QPs \eqref{QP gen} and \eqref{QP safe k} for all $x\in K$ and $x\in K\setminus \textnormal{int}(K_c)$, respectively. Then, the QPs \eqref{QP gen} and \eqref{QP safe k} are feasible for all $x\in K$ and $x\in K\setminus \textnormal{int}(K_c)$, respectively, $v^*, v_s^*$ are continuous on $\textnormal{int}(K)$ and $x\in \textnormal{int}(K\setminus \textnormal{int}(K_c))$, and the control input defined in \eqref{eq: switch input} with $\lambda(x) = v^*(x)$ and $k_s(x) = v_s^*(x)$ and $t_d = \hat t_d$, where $\hat  t_d$ is defined in \eqref{eq: t_d estimate}, solves Problem \ref{Problem 1} for all $x(0)\in  \textnormal{int}(K)$.
\end{Theorem}

\begin{proof}
Per Assumption \ref{assum: dot B bar c}, the set $K$ is a viability domain for the system \eqref{eq: cont affine}. Per Assumption \ref{assum: Sc set}, any sublevel set of $B$ in $K\setminus\textnormal{int}(K_{\bar c})$ is a viability domain for the system \eqref{eq: cont affine} under attack. Thus, feasibility of the QPs \eqref{QP gen} and \eqref{QP safe k} follows from \cite[Lemma 6]{garg2019prescribedTAC}. Per \cite[Theorem 1]{garg2019prescribedTAC}, the respective solutions of the QPs \eqref{QP gen} and \eqref{QP safe k} are continuous on $\textnormal{int}(K)$ and $\textnormal{int}(K\setminus\textnormal{int}(K_{\bar c}))$, respectively. Finally, since the set $K$ is compact, it follows from \cite[Lemma 7]{garg2019prescribedTAC} that the closed-loop trajectories are uniquely defined for all $t\geq 0$. The uniqueness of the closed-loop trajectories, Assumption \ref{assum: d bound} and feasibility of the QPs \eqref{QP gen} and \eqref{QP safe k} for all $x\in K$ and $x\in K\setminus\textnormal{int}(K_{\bar c})$ implies that all the conditions of Theorem \ref{thm: suff safety under attack} are satisfied with $\lambda$ defined as the solution of \eqref{QP gen} (i.e., $\lambda(x) = v^*(x)$) and $k_s$ as the solution of \eqref{QP safe k} (i.e., $k_s(x) = v_s^*(x)$). It follows that the set $\textnormal{int}(K)$ is forward invariant for the system \eqref{eq: cont affine}. 
\end{proof}

Thus, the QPs \eqref{QP gen} and \eqref{QP safe k} can be used to synthesize a nominal and a safe input for a system under attack. Next, we present a numerical case study involving an attack on one of the motors of a quadrotor and demonstrate how the proposed defense mechanism can save the quadrotor from crashing and keep it hovering at the desired altitude.  

\section{Case study}\label{sec: numerical}
We consider a simulation case study involving a quadrotor with an attack on one of its motors.\footnote{A video of the simulation is available at \url{https://tinyurl.com/3xzkute6} and the code is available at: \url{https://github.com/kunalgarg42/InputAttackRecovery}.} The quadrotor dynamics are given as (see \cite{lanzon2014flight,akhtar2013fault}): 
\begin{subequations}
\begin{align}
    \ddot x & = \frac{1}{m}\Big(\big(c(\phi)c(\psi)s(\theta)+s(\phi)s(\psi)\big)u_f-k_t\dot x\Big) \\
    \ddot y & = \frac{1}{m}\Big(\big(c(\phi)s(\psi)s(\theta)-s(\phi)c(\psi)\big)u_f-k_t\dot y\Big)\\
    \ddot z & = \frac{1}{m}\Big(c(\theta)c(\phi)u_f-mg-k_t\dot z\Big)\\
    \dot \phi & = p+qs(\phi)t(\theta)+rc(\phi)t(\theta)\\
    \dot \theta & = qc(\phi)-rs(\phi)\\
    \dot \psi & = \frac{1}{c(\theta)}\big(qs(\phi)+rc(\phi)\big)\\
    \dot p & = \frac{1}{I_{xx}}\Big(-k_rp-qr(I_{zz}-I_{yy})+\tau_p \Big)\\
    \dot q & = \frac{1}{I_{yy}}\big(-k_rq-pr(I_{xx}-I_{zz})+\tau_q \big)\\
    \dot r & = \frac{1}{I_{zz}}\big(-k_rr-pq(I_{yy}-I_{zz})+\tau_r \big),
\end{align}
\end{subequations}
where $m, I_{xx}, I_{yy}, I_{zz}, k_r, k_t>0$ are system parameters, $g = 9.8$ is the gravitational acceleration, $c(\cdot), s(\cdot), t(\cdot)$ denote $\cos(\cdot), \sin(\cdot), \tan(\cdot)$, respectively,  $(x, y, z)$ denote the position of the quadrotor, $(\phi, \theta, \psi)$ its Euler angles and $u = (u_f, \tau_p, \tau_q, \tau_r)$ the input vector consisting of thrust $u_f$ and moments $\tau_p, \tau_q, \tau_r$.
\begin{figure}[t]
	\centering
	\includegraphics[width=0.8\columnwidth,clip]{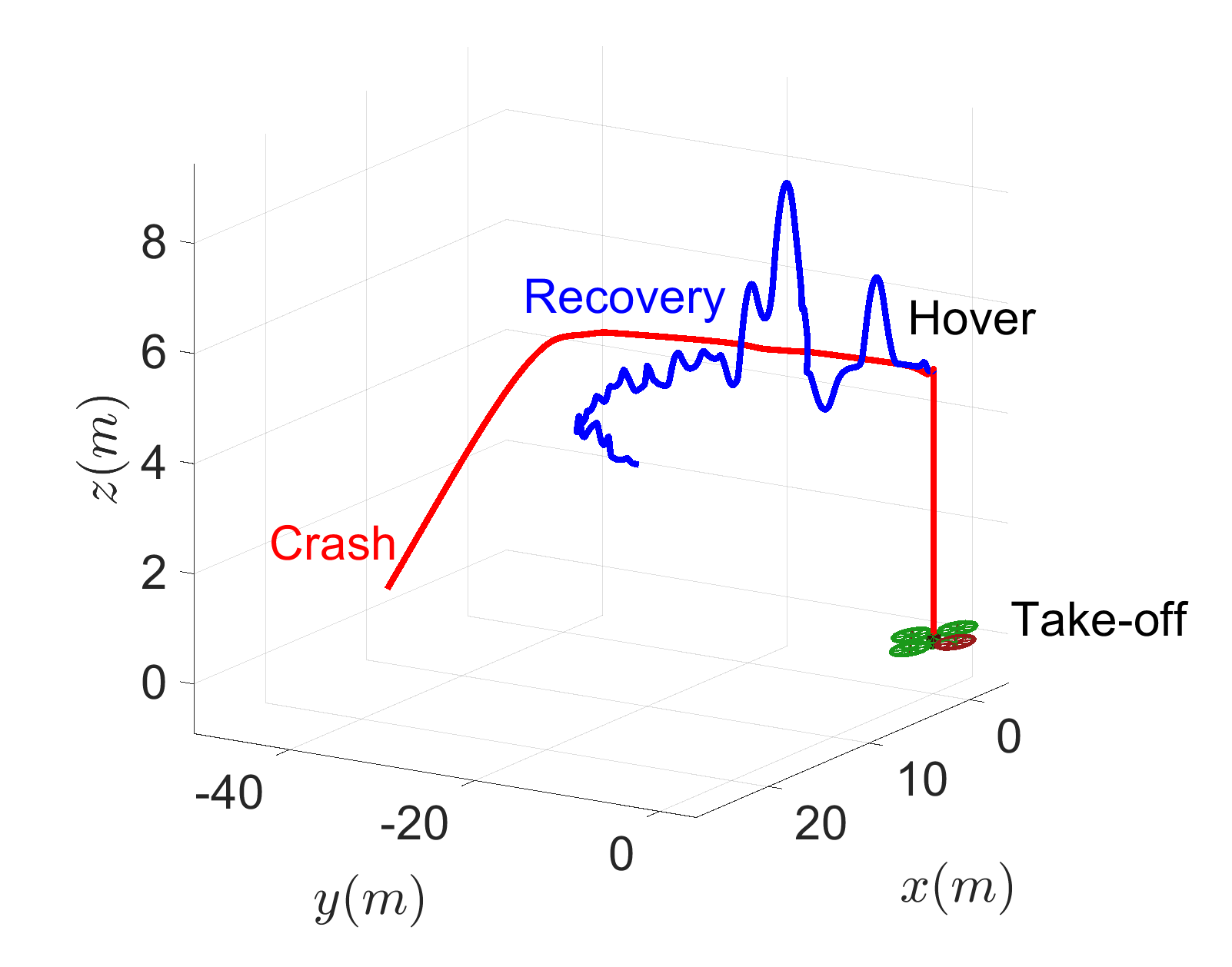}
	\caption{The closed-loop path traced by the quadrotor with the proposed detection mechanism (in blue) and without the detection mechanism (in red). The vulnerable motor is shown in red.}
	\label{fig: quad path}
\end{figure}
\begin{figure}[b]
	\centering
	\includegraphics[width=1\columnwidth,clip]{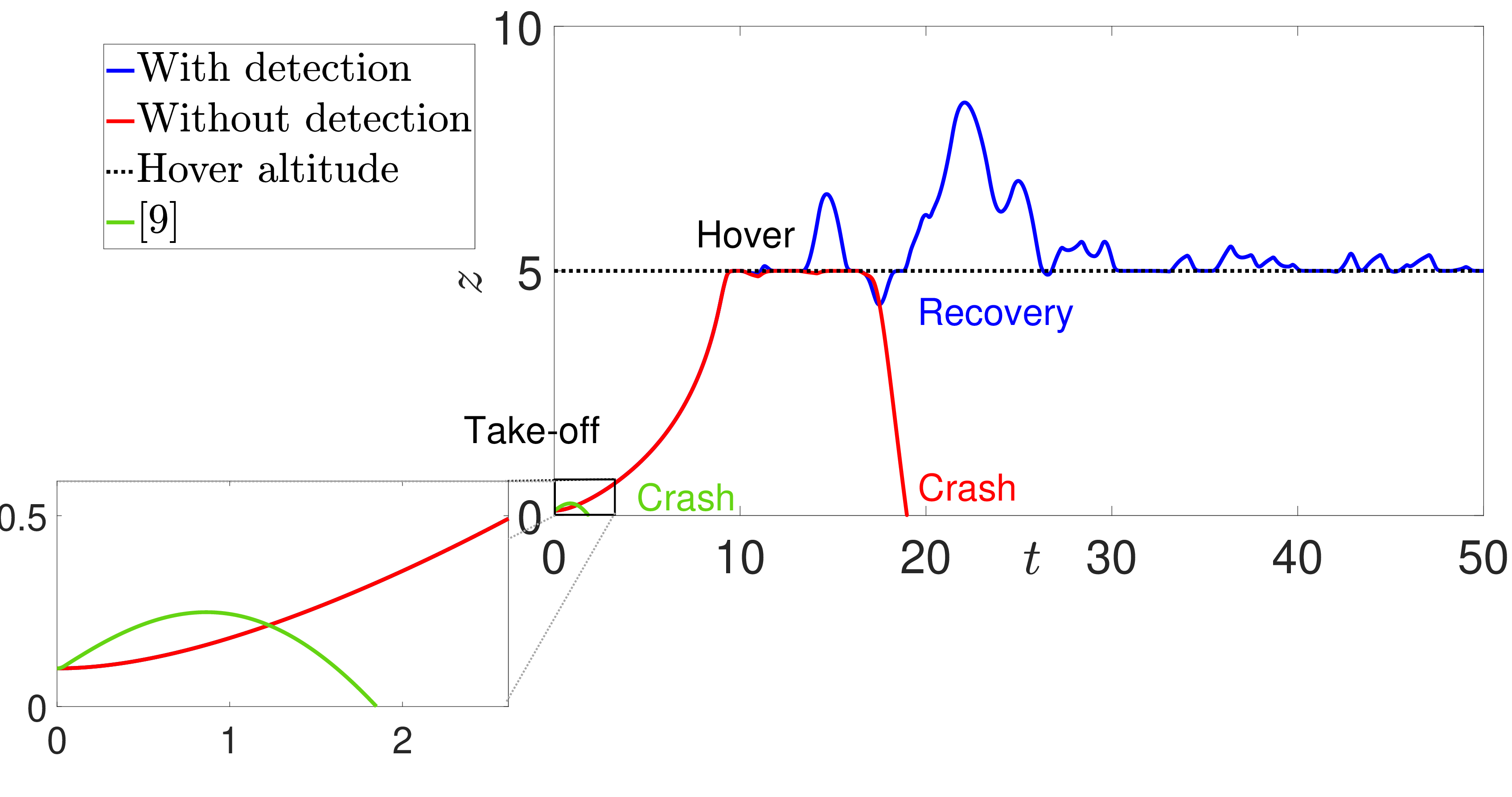}
	\caption{The $z-$coordinate of the closed-loop system with and without the detection mechanism. In the absence of the detection mechanism, the quadrotor crashes (i.e., $z = 0$ m). In the presence of the detection mechanism, the altitude remains close to the desired altitude $z = 5 m$ (shown by black line). The conservative approach in \cite{garg2021sampling}, resulting in crash even without an attack, is shown in green (see the inset plot).}
	\label{fig: xyz}
\end{figure}
\begin{figure}[t]
	\centering
	\includegraphics[width=1\columnwidth,clip]{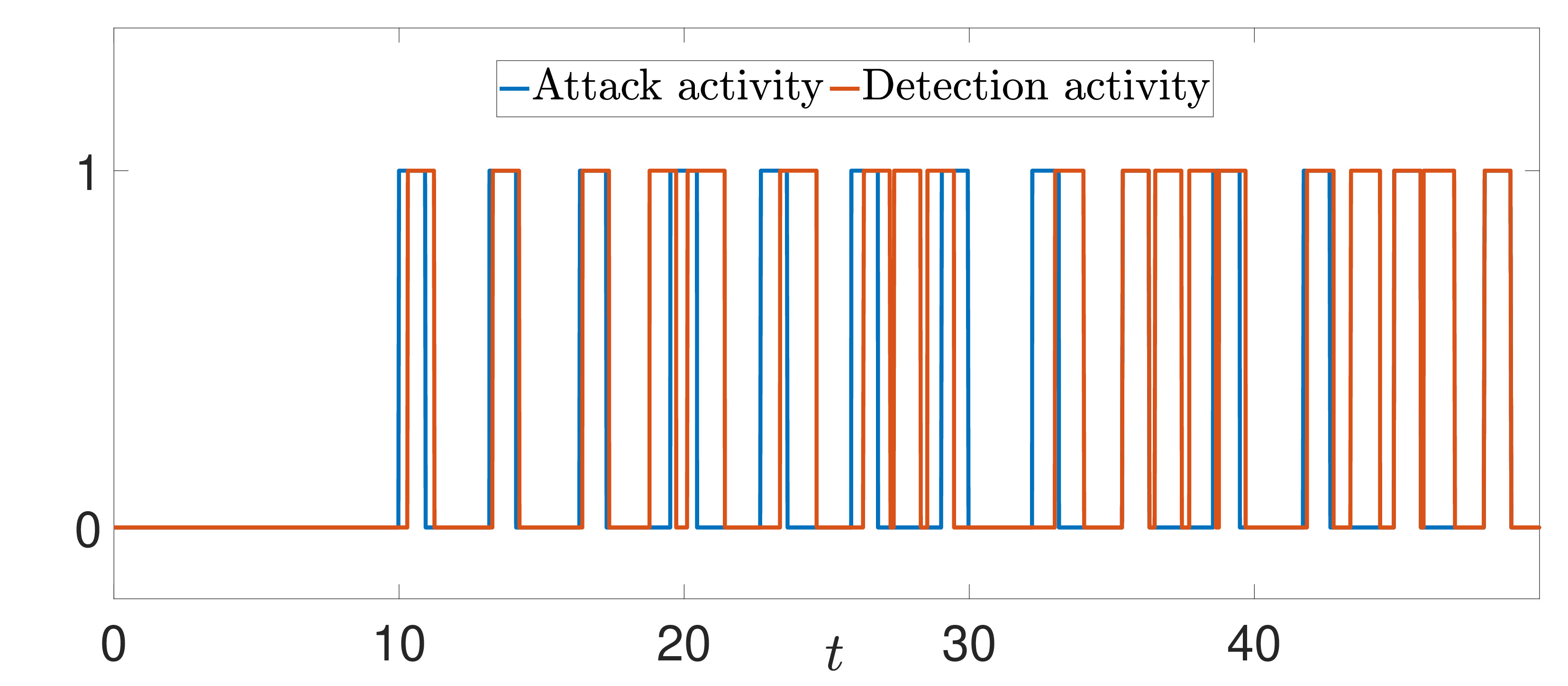}
	\caption{The attack (respectively, the detection) activity where 1 denotes that attack is active (respectively, flagged) and 0, that the attack is non-active (respectively, not flagged).} 
	\label{fig: attack sig}
\end{figure}
\begin{figure}[b]
	\centering
	\includegraphics[width=1\columnwidth, height =2.3in, clip]{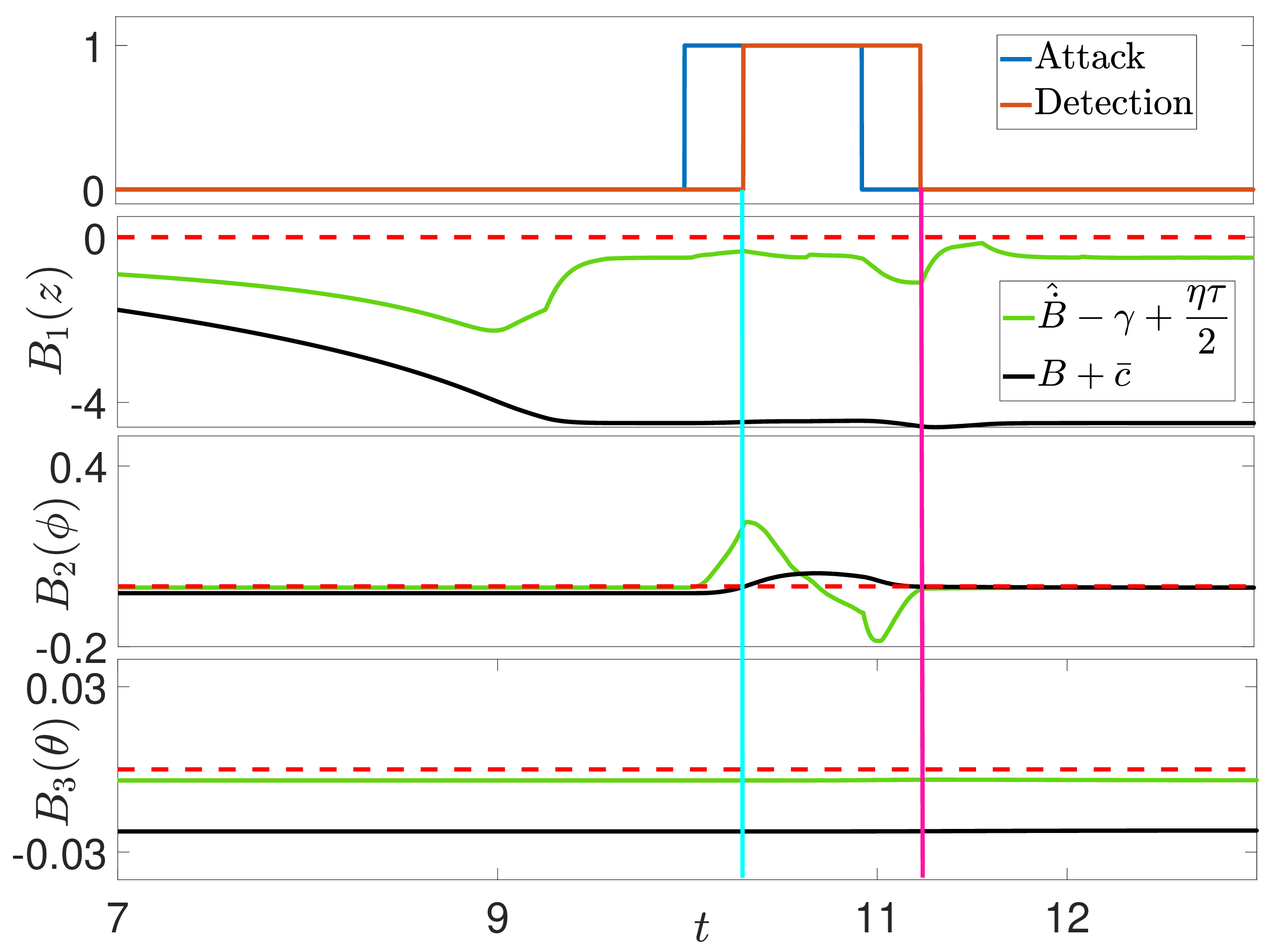}
\caption{The detection mechanism in action (vertical cyan line marks the beginning of the flagging and the vertical pink line, its ending). The attack is flagged per \eqref{eq: t_d estimate adapt} when $B(x(t))+\bar c = 0$ (shown in black line) and $\hat{\dot B}(t)-\gamma(t)+\frac{\eta\tau}{2} = 0$ (shown in green line). The first flag is raised at $t = 10$ second for $B = B_1(z)$. The mechanism keeps the system in the flagged mode for the maximum length of the attack, even if the attack is stopped.} 
	\label{fig: attack detect}
\end{figure}
The relation between the vector $u$ and the individual motor thrusts is given as 
\begin{align}
    \begin{bmatrix}u_f\\\tau_p\\\tau_q\\\tau_r\end{bmatrix} = \begin{bmatrix}1 & 1& 1& 1\\
    0& -l & 0 & l\\
    -l & 0 & l & 0\\
    d & -d  & d & -d
    \end{bmatrix} \begin{bmatrix}f_1\\ f_2\\f_3\\f_4\end{bmatrix},
\end{align}
where $f_i$ is the thrust generated by the $i-$th motor for $i\in \{1, 2, 3, 4\}$, $d, l>0$ are system parameters. We choose the system parameters for simulations as: $I_{xx} = I_{yy} = 0.177$ kg-$\textnormal{m}^2$, $I_{zz} = 0.344$ kg-$\textnormal{m}^2$, $m = 4.493$ kg, $l = 0.1$ m, $d = 0.0024\;\textnormal{m}$, $k_t = 1$ and $k_r = 1.5$ (see \cite{akhtar2013fault}). Furthermore, we consider the bound on each motor given as $|f_i|\leq 27.7$ N for $i \in \{1, 2, 3, 4\}$. We use $\tau = 10^{-3}$ for approximation of $\dot B$. Without loss of generality, we assume that motor \#4 is vulnerable. Note that under an attack, the input-thrust relation reads:
\begin{align}
    \begin{bmatrix}u_f\\\tau_p\\\tau_q\\\tau_r\end{bmatrix} = \begin{bmatrix}1 & 1& 1\\
    0& -l & 0\\
    -l & 0 & l\\
    d & -d  & d 
    \end{bmatrix} \begin{bmatrix}f_1\\ f_2\\f_3\end{bmatrix},
\end{align}
It is not possible to keep all the inputs $(u_f, \tau_p, \tau_q, \tau_r)$ close to its desired value simultaneously under an attack on motor \#4. Thus, we focus on designing a control law to maintain the desired altitude of the quadrotor (through $u_f$) and minimize its oscillations (through $(\tau_p, \tau_q)$). It implies that $\tau_r$ will not be matched with its desired value to control the yaw angle $\psi$, resulting in an uncontrolled yaw angle increase.

We choose the control objective to make the quadrotor hover at location $(0, 0, 5)$, starting from $(0, 0, 0.2)$. Based on the above observation and the fact that $\psi$ does not contribute in changing the altitude of the quadrotor, the safety constraints are to keep the angles $(\phi, \theta)$ in a given bounded range, i.e., $|\phi|\leq \phi_M, |\theta|\leq \theta_M$, for some $\phi_M, \theta_M>0$, and to keep the quadrotor above the ground, i.e., $z>0$. Thus, the safe set is defined as $K = \Big\{(\phi, \theta, z)\; |\; |\phi|\leq \phi_M, |\theta|\leq \theta_M, z\leq -\epsilon\Big\}$. We choose $\phi_M = \theta_M = 0.3$ and $\epsilon = 0.02$. The maximum length of the attack is randomly chosen as $\overline T = 0.934$ seconds and the period of no attack is chosen as $T_{na} = 2.238$ seconds.

\begin{figure}[t]
	\centering
	\includegraphics[width=1\columnwidth,clip]{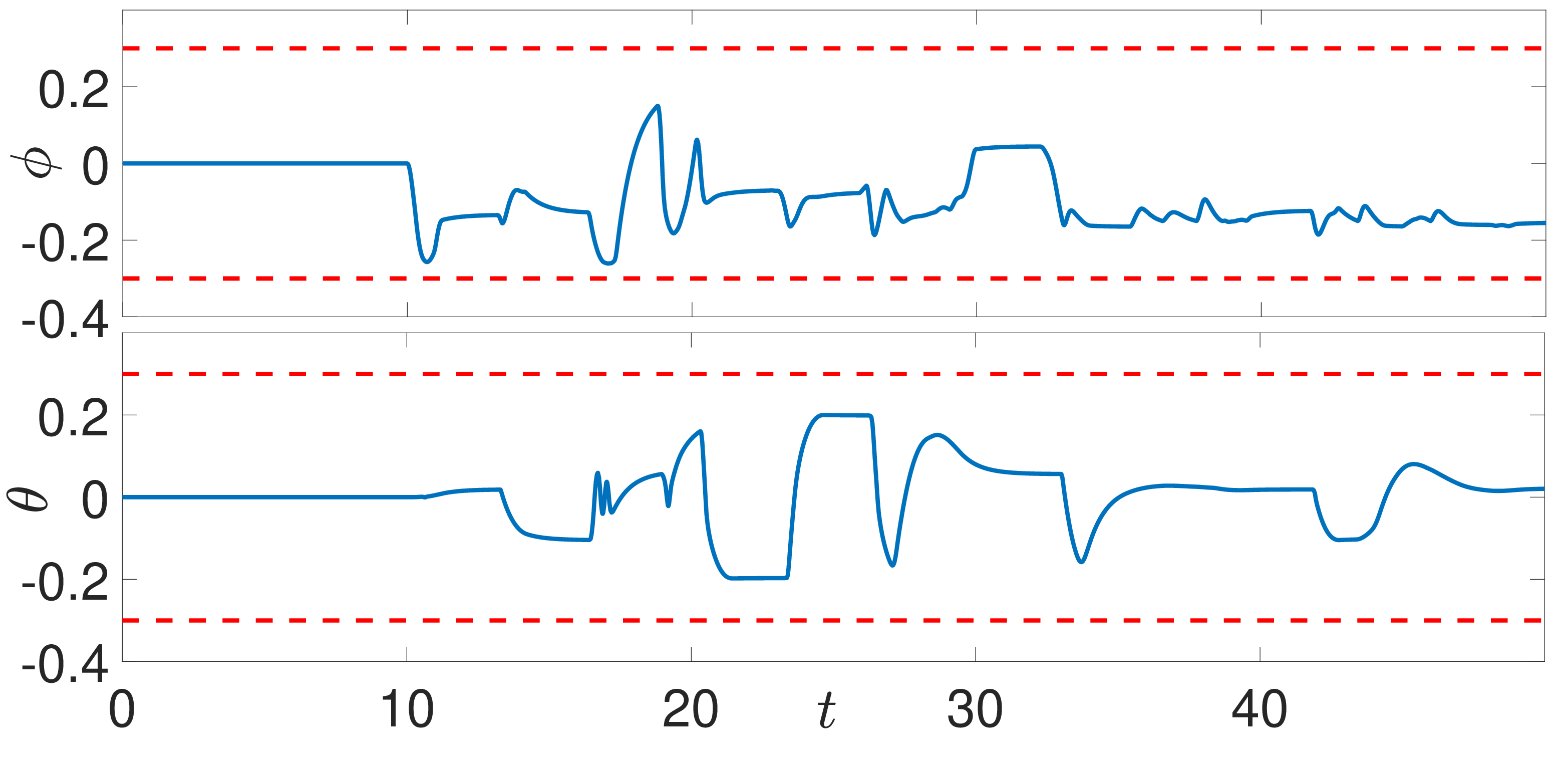}
	\caption{Euler angles $(\phi, \theta)$ of the closed-loop system. The safety constraints $|\phi|\leq \phi_M$ and $|\theta|\leq \theta_M$ are satisfied at all times.}
	\label{fig: eul ang}
\end{figure}

\begin{figure}[b]
	\centering
	\includegraphics[width=1\columnwidth, height = 2.3in, clip]{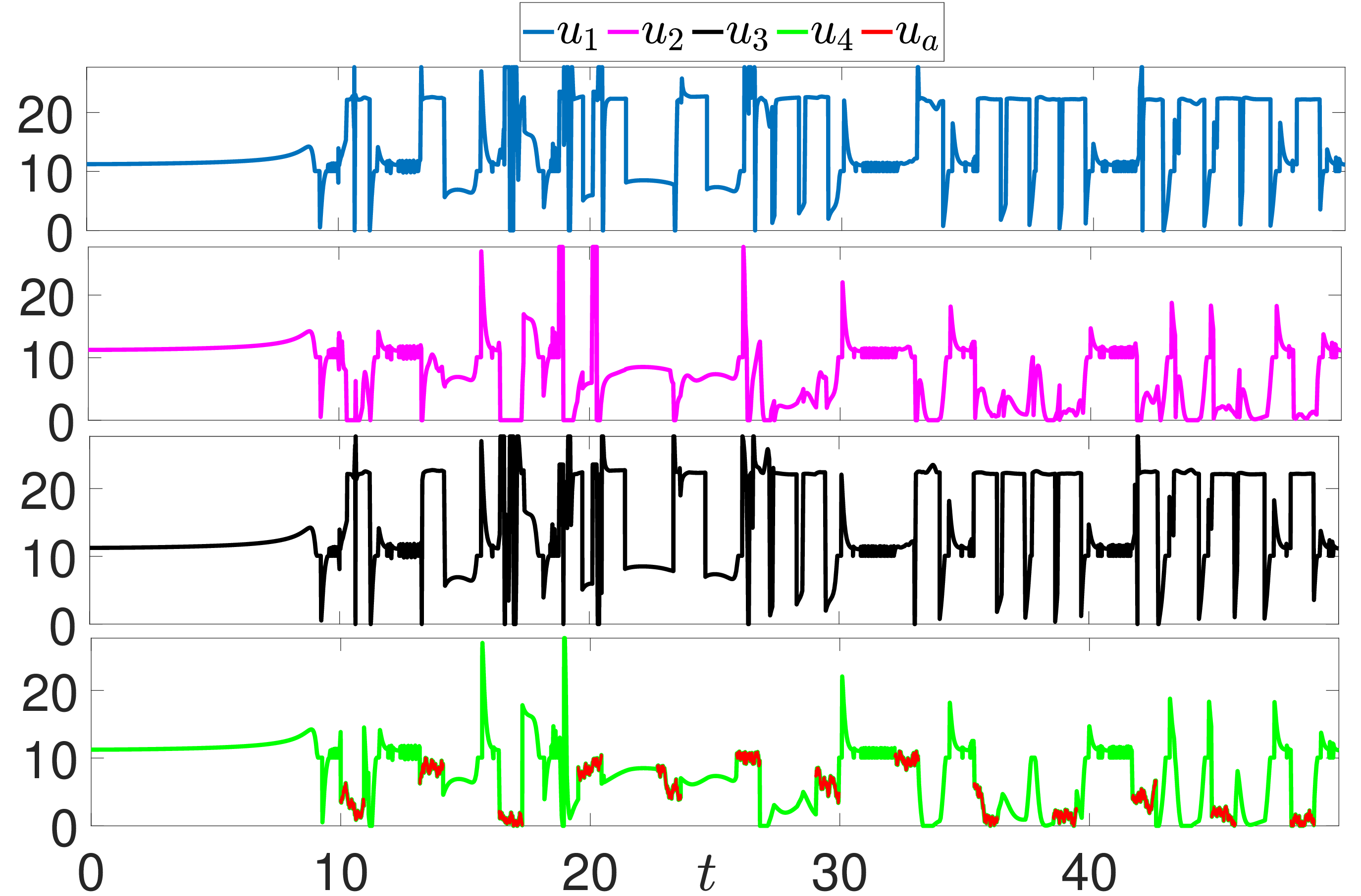}
	\caption{Thrust $f_i$ of each motor. The thrust of motor 4 under attack is shown in red. The switch in the rest of the motors is clearly seen when an attack is flagged.}
	\label{fig: input sig}
\end{figure}

The barrier functions used for enforcing safety are $B_1(z) = -z+0.02$, $B_2(\phi) = |\phi|^2-\phi_M^2$ and $B_3(\theta) = |\theta|^2-\theta_M^2$. The parameters $\bar \delta,\bar c$ for detection are $\bar\delta = 0.1, \bar c = \frac{1}{4}(0.3)^2$. Figure \ref{fig: quad path} shows the closed-loop path traced by the quadrotor. Figure \ref{fig: xyz} plots the position coordinates $(x,y,z)$. The safety constraint $z\leq 0$ is satisfied at all times, and the quadrotor is able to hover at an altitude $z = 5$ m. Figure \ref{fig: attack sig} shows the attack and the detection signal. It can be seen that detection has a non-zero delay during some attacks, and zero delay during some attacks. It can also be seen that some of the attacks are not detected, as they do not fall into the category of \textit{adversarial} attack per Definition \ref{def: advers attack}. Figure \ref{fig: attack detect} illustrates the detection mechanism in action. The attack is flagged according to \eqref{eq: t_d estimate adapt} and remains flagged for the duration $\overline T$. The bound $|f_i|\leq 27.7$ N is satisfied for each motor at all times.  The vulnerable motor is highlighted in green. Figure \ref{fig: eul ang} plots the Euler angles $(\phi, \theta)$. It can be seen that the safety constraint $|\phi|\leq 0.3$ and $|\theta|\leq 0.3$ is satisfied at all times. Finally, Figure \ref{fig: input sig} plots the thrust for each motor under nominal conditions as well as under attack.


Thus, the proposed scheme can successfully detect an attack on a quadrotor motor before the quadrotor crashes. Furthermore, the designed safe input can keep the quadrotor in the safe zone even under attack, thus demonstrating a successful recovery after detection. The conservative approach in \cite{garg2021sampling}, which assumes that the rotor $\#$4 is constantly under attack, fails to keep the quadrotor from crashing even when there is no attack (see Figure \ref{fig: xyz}). In contrast, the proposed approach is non-conservative and reacts to an adversarial attack, thereby not interfering with the system's nominal functionality.


The simulation results illustrate that the approach in \cite{garg2021sampling} is too conservative for the considered example, and that the chosen initial condition does not satisfy the requirements of the framework in \cite{garg2021sampling}. The addition of the detection mechanism removes this conservatism and results in improved system performance, even if there is no attack. 

It is also important to note the difference between fault-tolerant control (FTC) (see e.g. \cite{lanzon2014flight,akhtar2013fault} in the context of quadrotor control). The control scheme under the FTC paradigm assumes that a subset of actuators have failed and are not operating nominally. Furthermore, the focus of FTC-based schemes is to control the system with the available \textit{non-faulty} actuators. In contrast, the focus of the proposed scheme is to not only control the system with the \textit{non-vulnerable} actuators but also to design them in a way that for all possible \textit{attacked} signals, the system is still safe. Thus, the proposed method is robust against any random actuator signal.  
 
\section{Conclusion}\label{sec: conclusion}
We presented a novel attack-detection scheme based on the control Barrier function. In addition, we introduced an online QP-based formulation to design a recovery controller that prevents the system from violating the safety specification.  Our formulation is adaptive, in the sense that the further away the system is from violating safety our recovery controller focuses on performance rather than safety; however, if the system keeps approaching the safety limit, our adaptive mechanism switches to a recovery controller to counteract the potential attack. We demonstrated the efficacy of the proposed method on a simulation example involving an attack on a quadrotor motor. 


This work opens up a line of research on non-conservative control design for CPS security with provable guarantees. Provable safety guarantees when the system sensors are under attack is still an open problem. Future work involves studying more general attacks on CPS, such as attacks on system sensors and simultaneous attacks on system sensors and actuators. As noted in Remark 1, our future investigation also includes studying methods of estimating the time when the attack has stopped.


\bibliographystyle{IEEEtran}
\bibliography{myreferences}

\begin{thebibliography}{10}
\providecommand{\url}[1]{#1}
\csname url@samestyle\endcsname
\providecommand{\newblock}{\relax}
\providecommand{\bibinfo}[2]{#2}
\providecommand{\BIBentrySTDinterwordspacing}{\spaceskip=0pt\relax}
\providecommand{\BIBentryALTinterwordstretchfactor}{4}
\providecommand{\BIBentryALTinterwordspacing}{\spaceskip=\fontdimen2\font plus
\BIBentryALTinterwordstretchfactor\fontdimen3\font minus
  \fontdimen4\font\relax}
\providecommand{\BIBforeignlanguage}[2]{{%
\expandafter\ifx\csname l@#1\endcsname\relax
\typeout{** WARNING: IEEEtran.bst: No hyphenation pattern has been}%
\typeout{** loaded for the language `#1'. Using the pattern for}%
\typeout{** the default language instead.}%
\else
\language=\csname l@#1\endcsname
\fi
#2}}
\providecommand{\BIBdecl}{\relax}
\BIBdecl

\bibitem{cardenascyber}
\BIBentryALTinterwordspacing
A.~Cardenas, ``Cyber-physical systems security knowledge area issue.'' {T}he
  Cyber Security Body Of Knowledge. [Online]. Available:
  \url{https://www.cybok.org/media/downloads/Cyber-Physical_Systems_Security_issue_1.0.pdf}
\BIBentrySTDinterwordspacing

\bibitem{chen2018learning}
Y.~Chen, C.~M. Poskitt, and J.~Sun, ``Learning from mutants: Using code
  mutation to learn and monitor invariants of a cyber-physical system,'' in
  \emph{2018 IEEE Symposium on Security and Privacy (SP)}.\hskip 1em plus 0.5em
  minus 0.4em\relax IEEE, 2018, pp. 648--660.

\bibitem{choi2018}
H.~Choi, W.-C. Lee, Y.~Aafer, F.~Fei, Z.~Tu, X.~Zhang, D.~Xu, and X.~Deng,
  ``Detecting attacks against robotic vehicles: A control invariant approach,''
  in \emph{Proceedings of the 2018 ACM SIGSAC Conference on Computer and
  Communications Security}, ser. CCS '18.\hskip 1em plus 0.5em minus
  0.4em\relax New York, NY, USA: ACM, 2018, pp. 801--816.

\bibitem{fengsystematic}
C.~Feng, V.~R. Palleti, A.~Mathur, and D.~Chana, ``A systematic framework to
  generate invariants for anomaly detection in industrial control systems,'' in
  \emph{2019 Network and Distributed System Security Symposium {(NDSS)}.}

\bibitem{renganathan2020distributionally}
V.~Renganathan, N.~Hashemi, J.~Ruths, and T.~H. Summers, ``Distributionally
  robust tuning of anomaly detectors in cyber-physical systems with stealthy
  attacks,'' in \emph{2020 American Control Conference (ACC)}.\hskip 1em plus
  0.5em minus 0.4em\relax IEEE, 2020, pp. 1247--1252.

\bibitem{urbina2016limiting}
D.~I. Urbina, J.~A. Giraldo, A.~A. Cardenas, N.~O. Tippenhauer, J.~Valente,
  M.~Faisal, J.~Ruths, R.~Candell, and H.~Sandberg, ``Limiting the impact of
  stealthy attacks on industrial control systems,'' in \emph{Proceedings of the
  2016 ACM SIGSAC Conference on Computer and Communications Security}, 2016,
  pp. 1092--1105.

\bibitem{al2018cyber}
M.~N. Al-Mhiqani, R.~Ahmad, W.~Yassin, A.~Hassan, Z.~Z. Abidin, N.~S. Ali, and
  K.~H. Abdulkareem, ``Cyber-security incidents: a review cases in
  cyber-physical systems,'' \emph{Int. J. Adv. Comput. Sci. Appl}, no.~1, pp.
  499--508, 2018.

\bibitem{ames2017control}
A.~D. Ames, X.~Xu, J.~W. Grizzle, and P.~Tabuada, ``Control barrier function
  based quadratic programs for safety critical systems,'' \emph{IEEE
  Transactions on Automatic Control}, vol.~62, no.~8, pp. 3861--3876, 2017.

\bibitem{garg2021sampling}
K.~Garg, A.~A. Cardenas, and R.~G. Sanfelice, ``Sampling-based computation of
  viability domain to prevent safety violations by attackers,'' in \emph{2022
  IEEE Conference on Control Technology and Applications (CCTA)}.\hskip 1em
  plus 0.5em minus 0.4em\relax IEEE, 2022, pp. 720--725.

\bibitem{wang2018permissive}
L.~Wang, D.~Han, and M.~Egerstedt, ``Permissive barrier certificates for safe
  stabilization using sum-of-squares,'' in \emph{2018 Annual American Control
  Conference (ACC)}.\hskip 1em plus 0.5em minus 0.4em\relax IEEE, 2018, pp.
  585--590.

\bibitem{choi2021robust}
J.~J. Choi, D.~Lee, K.~Sreenath, C.~J. Tomlin, and S.~L. Herbert, ``Robust
  control barrier–value functions for safety-critical control,'' in
  \emph{2021 60th IEEE Conference on Decision and Control (CDC)}, 2021.

\bibitem{garg2022control}
K.~Garg, R.~G. Sanfelice, and A.~A. Cardenas, ``Control barrier function-based
  attack-recovery with provable guarantees,'' in \emph{2022 IEEE 61st
  Conference on Decision and Control (CDC)}.\hskip 1em plus 0.5em minus
  0.4em\relax IEEE, 2022, pp. 4808--4813.

\bibitem{clark2020control}
A.~Clark, Z.~Li, and H.~Zhang, ``Control barrier functions for safe cps under
  sensor faults and attacks,'' in \emph{2020 59th IEEE Conference on Decision
  and Control (CDC)}.\hskip 1em plus 0.5em minus 0.4em\relax IEEE, 2020, pp.
  796--803.

\bibitem{ramasubramanian2019linear}
B.~Ramasubramanian, L.~Niu, A.~Clark, L.~Bushnell, and R.~Poovendran, ``Linear
  temporal logic satisfaction in adversarial environments using secure control
  barrier certificates,'' in \emph{International Conference on Decision and
  Game Theory for Security}.\hskip 1em plus 0.5em minus 0.4em\relax Springer,
  2019, pp. 385--403.

\bibitem{giraldo2020daria}
J.~Giraldo, S.~H. Kafash, J.~Ruths, and A.~A. Cardenas, ``Daria: Designing
  actuators to resist arbitrary attacks against cyber-physical systems,'' in
  \emph{2020 IEEE European Symposium on Security and Privacy (EuroS\&P)}.\hskip
  1em plus 0.5em minus 0.4em\relax IEEE, 2020, pp. 339--353.

\bibitem{pasqualetti2013attack}
F.~Pasqualetti, F.~D{\"o}rfler, and F.~Bullo, ``Attack detection and
  identification in cyber-physical systems,'' \emph{IEEE Transactions on
  Automatic Control}, vol.~58, no.~11, pp. 2715--2729, 2013.

\bibitem{kafash2018constraining}
S.~H. Kafash, J.~Giraldo, C.~Murguia, A.~A. Cardenas, and J.~Ruths,
  ``Constraining attacker capabilities through actuator saturation,'' in
  \emph{2018 Annual American Control Conference}.\hskip 1em plus 0.5em minus
  0.4em\relax IEEE, 2018, pp. 986--991.

\bibitem{aubin2012differential}
J.-P. Aubin and A.~Cellina, \emph{Differential inclusions: set-valued maps and
  viability theory}.\hskip 1em plus 0.5em minus 0.4em\relax Springer Science \&
  Business Media, 2012, vol. 264.

\bibitem{chai2018forward}
J.~Chai and R.~G. Sanfelice, ``Forward invariance of sets for hybrid dynamical
  systems (part i),'' \emph{IEEE Transactions on Automatic Control}, vol.~64,
  no.~6, pp. 2426--2441, 2018.

\bibitem{maghenem2021sufficient}
M.~Maghenem and R.~G. Sanfelice, ``Sufficient conditions for forward invariance
  and contractivity in hybrid inclusions using barrier functions,''
  \emph{Automatica}, vol. 124, p. 109328, 2021.

\bibitem{chai2020forward}
J.~Chai and R.~G. Sanfelice, ``Forward invariance of sets for hybrid dynamical
  systems (part ii),'' \emph{IEEE Transactions on Automatic Control}, vol.~66,
  no.~1, pp. 89--104, 2020.

\bibitem{cominetti1990generalized}
R.~Cominetti and R.~Correa, ``A generalized second-order derivative in
  nonsmooth optimization,'' \emph{SIAM Journal on Control and Optimization},
  vol.~28, no.~4, pp. 789--809, 1990.

\bibitem{goebel2012hybrid}
R.~Goebel, R.~G. Sanfelice, and A.~R. Teel, ``Hybrid dynamical systems,'' in
  \emph{Hybrid dynamical systems}.\hskip 1em plus 0.5em minus 0.4em\relax
  Princeton University Press, 2012.

\bibitem{clarke1995proximal}
F.~H. Clarke, R.~J. Stern, and P.~R. Wolenski, ``Proximal smoothness and the
  lower-c2 property,'' \emph{J. Convex Anal}, vol.~2, no. 1-2, pp. 117--144,
  1995.

\bibitem{rockafellar1981favorable}
R.~T. Rockafellar, ``Favorable classes of lipschitz continuous functions in
  subgradient optimization,'' 1981.

\bibitem{poliquin1996generalized}
R.~A. Poliquin and R.~T. Rockafellar, ``Generalized hessian properties of
  regularized nonsmooth functions,'' \emph{SIAM Journal on Optimization},
  vol.~6, no.~4, pp. 1121--1137, 1996.

\bibitem{wintz2022global}
P.~K. Wintz, R.~G. Sanfelice, and J.~Hespanha, ``Global asymptotic stability of
  nonlinear systems while exploiting properties of uncertified feedback
  controllers via opportunistic switching,'' in \emph{Proceedings of the
  American Control Conference}, June 2022.

\bibitem{sanfelice2007invariance}
R.~G. Sanfelice, R.~Goebel, and A.~R. Teel, ``Invariance principles for hybrid
  systems with connections to detectability and asymptotic stability,''
  \emph{IEEE Transactions on Automatic Control}, vol.~52, no.~12, pp.
  2282--2297, 2007.

\bibitem{garg2019prescribedTAC}
K.~Garg, E.~Arabi, and D.~Panagou, ``Fixed-time control under spatiotemporal
  and input constraints: A quadratic programming based approach,''
  \emph{Automatica}, vol. 141, p. 110314, 2022.

\bibitem{lanzon2014flight}
A.~Lanzon, A.~Freddi, and S.~Longhi, ``Flight control of a quadrotor vehicle
  subsequent to a rotor failure,'' \emph{Journal of Guidance, Control, and
  Dynamics}, vol.~37, no.~2, pp. 580--591, 2014.

\bibitem{akhtar2013fault}
A.~Akhtar, S.~L. Waslander, and C.~Nielsen, ``Fault tolerant path following for
  a quadrotor,'' in \emph{52nd IEEE Conference on Decision and Control}.\hskip
  1em plus 0.5em minus 0.4em\relax IEEE, 2013, pp. 847--852.

\end{thebibliography}

\begin{biography}[{\includegraphics[width=1in,height=1.25in,clip,keepaspectratio]{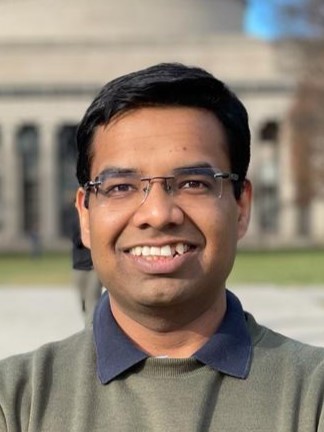}}]{Kunal Garg} received his Master of Engineering and PhD degrees in Aerospace Engineering from University of Michigan in 2019 and 2021, and his Bachelor of Technology degree in Aerospace Engineering from the Indian Institute of Technology, Mumbai, India in 2016. He is currently an assistant professor in the Mechanical and Aerospace Engineering Program at the School for Engineering of Matter, Transport, and Energy at Arizona State University. Previously, he was a postdoctoral associate in the Department of Aeronautics and Astronautics at Massachusetts Institute of Technology and before that, with the Department of Electrical Engineering and Computer Science at the University of California, Santa Cruz. His research interests include robust multi-agent path planning, switched and hybrid system-based analysis and control synthesis for multi-agent coordination, finite- and fixed-time stability of dynamical systems with applications to control synthesis for spatiotemporal specifications, and continuous-time optimization. He is a member of the IEEE.
\end{biography}

\begin{biography}[{\includegraphics[width=1in,height=1.25in,clip,keepaspectratio]{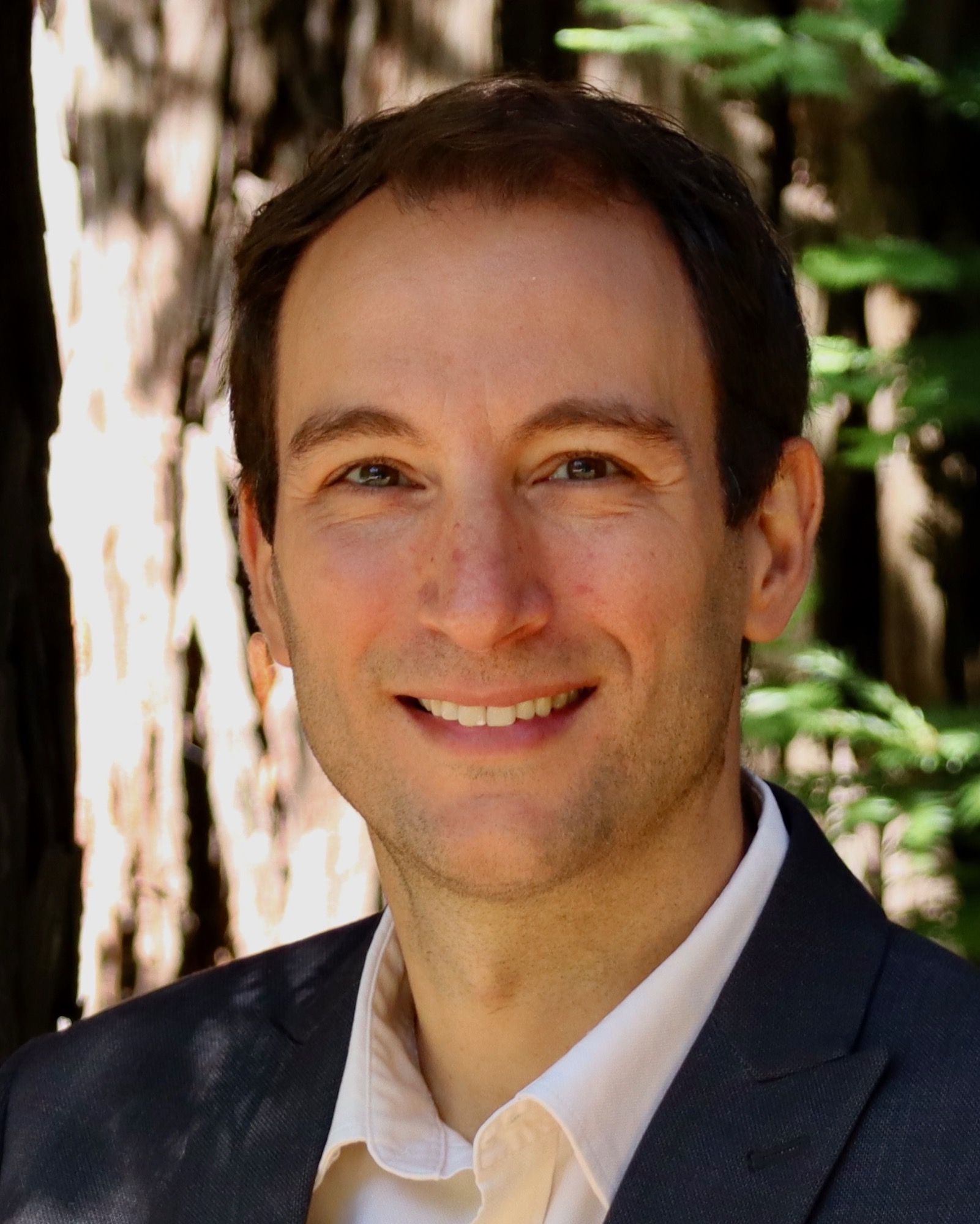}}]{Ricardo G. Sanfelice} received the B.S. degree in Electronics Engineering from the Universidad de Mar del Plata, Buenos Aires, Argentina, in 2001, and the M.S. and Ph.D. degrees in Electrical and Computer Engineering from the University of California, Santa Barbara, in 2004 and 2007, respectively. In 2007 and 2008, he held postdoctoral positions at the Laboratory for Information and Decision Systems at the Massachusetts Institute of Technology and at the Centre Automatique et Systèmes at the École de Mines de Paris. In 2009, he joined the faculty of the Department of Aerospace and Mechanical Engineering at the University of Arizona, Tucson, where he was an Assistant Professor. In 2014, he joined the University of California, Santa Cruz, where he is currently Professor and Chair in the Department of Electrical and Computer Engineering. Prof. Sanfelice is the recipient of the 2013 SIAM Control and Systems Theory Prize, the National Science Foundation CAREER award, the Air Force Young Investigator Research Award, the 2010 IEEE Control Systems Magazine Outstanding Paper Award, and the 2020 Test-of-Time Award from the Hybrid Systems: Computation and Control Conference. He is Associate Editor for Automatica, Communicating Editor for the Journal of Nonlinear Science, and a Fellow of the IEEE. His research interests are in modeling, stability, robust control, observer design, and simulation of nonlinear and hybrid systems with applications to power systems, aerospace, and biology.
\end{biography}

\begin{biography}[{\includegraphics[width=1in,height=1.25in,clip,keepaspectratio]{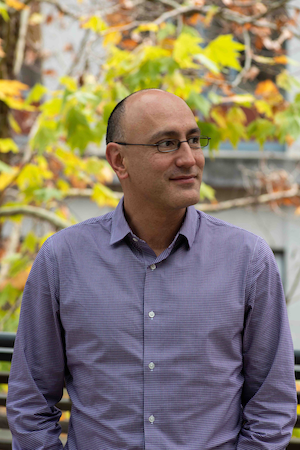}}]{Alvaro A. Cardenas} is a Professor of Computer Science and Engineering at the University of California, Santa Cruz. Before joining UCSC he was the Eugene McDermott Associate Professor of Computer Science at the University of Texas at Dallas, a postdoctoral scholar at the University of California, Berkeley, and a research staff member at Fujitsu Laboratories. He holds M.S. and Ph.D. degrees from the University of Maryland, College Park, and a B.S. from Universidad de Los Andes in Colombia. His research interests focus on cyber-physical systems security, including autonomous vehicles, drones, smart home devices, and SCADA systems controlling the power grid and other critical infrastructures. He is the recipient of the NSF CAREER award, the 2018 faculty excellence in research award from the Erik Johnson School of Engineering and Computer Science, the Eugene McDermott Fellow Endowed Chair at UTD, and the Distinguished Service Award from the IEEE Computer Society Technical Committee on Security and Privacy. He has also received best paper awards from various venues, including the ACM CPS \& IoT Security Workshop, IEEE Smart Grid Communications Conference and the U.S. Army Research Conference. 
\end{biography}

\end{document}